\RequirePackage{fix-cm}

\documentclass[smallextended,envcountsame]{svjour3}

\smartqed

\usepackage{times}

\usepackage{clrscode3e}
\usepackage{latexsym}
\usepackage{hyperref,todonotes}
\usepackage{amssymb,amsfonts,amsmath}
\usepackage{color}
\usepackage{tikz}
\usepackage{color}
\usepackage{comment}
\usepackage{enumitem}
\usepackage{epstopdf}
\usepackage{graphicx}
\usepackage{boxedminipage}

\newcommand{\bigoh}{\mathcal{O}}

\newcommand{\cB}{{\cal{B}}}

\newcommand{\cQ}{\mathcal{Q}}

\newcommand{\cT}{\mathcal{T}}

\newcommand{\cA}{\cal{A}}

\usetikzlibrary{shapes}
\usetikzlibrary{calc}
\usepackage{thm-restate}

\newcommand{\defproblem}[3]{
  \vspace{1mm}
\noindent\fbox{
  \begin{minipage}{0.96\textwidth}
  \begin{tabular*}{\textwidth}{@{\extracolsep{\fill}}lr} #1  &  \\ \end{tabular*}
  {\bf{Input:}} #2  \\
  {\bf{Question:}} #3
  \end{minipage}
  }
  \vspace{1mm}
}

\newcommand{\fa}{$f$-factor\ }

\newcommand{\cgfa}{\textsc{Connected $g$-Bounded $f$-factor}}

\newcommand{\mac}{minimal alternating circuit\ }
\newcommand{\macs}{minimal alternating circuits\ }

\renewcommand{\P}{\ensuremath{\mathsf{P}}}

\newcommand{\ceil}[1]{\lceil #1 \rceil}

\newcommand{\polylog}[1]{\mathsf{polylog}(#1)}

\newcommand{\NPC}{\ensuremath{\mathsf{NPC}}}{}
\newcommand{\NP}{\ensuremath{\mathsf{NP}}}{}
\newcommand{\QP}{\ensuremath{\mathsf{QP}}}{}

\renewcommand{\P}{\ensuremath{\mathsf{P}}}

\makeatletter
\def\renewtheorem#1{%
  \expandafter\let\csname#1\endcsname\relax
  \expandafter\let\csname c@#1\endcsname\relax
  \gdef\renewtheorem@envname{#1}
  \renewtheorem@secpar
}
\def\renewtheorem@secpar{\@ifnextchar[{\renewtheorem@numberedlike}{\renewtheorem@nonumberedlike}}
\def\renewtheorem@numberedlike[#1]#2{\newtheorem{\renewtheorem@envname}[#1]{#2}}
\def\renewtheorem@nonumberedlike#1{  
\def\renewtheorem@caption{#1}
\edef\renewtheorem@nowithin{\noexpand\newtheorem{\renewtheorem@envname}{\renewtheorem@caption}}
\renewtheorem@thirdpar
}
\def\renewtheorem@thirdpar{\@ifnextchar[{\renewtheorem@within}{\renewtheorem@nowithin}}
\def\renewtheorem@within[#1]{\renewtheorem@nowithin[#1]}
\makeatother

\newtheorem{Fact}[theorem]{Fact}
\newtheorem{observation}[theorem]{Observation}

\newcommand{\Nat}{\mathbb{N}}

\usepackage{chngcntr}

\makeatletter
\renewenvironment{proof}[1][\relax]{\par
  \normalfont \topsep6\p@\@plus6\p@\relax
  \trivlist
  \item[\hskip\labelsep\itshape
    \ifx#1\relax \proofname\else\proofname{} of #1\fi\@addpunct{.}]\ignorespaces
}{%
}\makeatother

\begin{document}
\title{On the Complexity Landscape of  Connected $f$-Factor Problems} 
\author{R. Ganian$^1$ \and N. S. Narayanaswamy$^2$ \and S. Ordyniak$^3$ \and C. S. Rahul$^4$ \and M. S. Ramanujan$^1$}

\authorrunning{Ganian, Narayanaswamy, Ordyniak, Rahul, Ramanujan}

\institute{$^1$Algorithms and Complexity Group, TU Wien, Vienna, Austria\\
  \texttt{rganian@gmail.com, ramanujan@ac.tuwien.ac.at}\\
$^2$Indian Institute of Technology Madras, Chennai, India\\
\texttt{swamy@cse.iitm.ac.in}\\
$^3$Department of Computer Science, University of Sheffield, Sheffield, UK\\
  \texttt{sordynian@gmail.com}\\
$^4$Faculty of Mathematics, Informatics and Mechanics, University of Warsaw, Poland\\
  \texttt{rahulcs@mimuw.edu.pl}  
  }

\maketitle
\begin{abstract}
Let $G$ be an undirected simple graph having $n$ vertices and let $f$ be a function defined to be $f:V(G)\rightarrow \{0,\dots, n-1\}$. An $f$-factor of $G$ is a spanning subgraph $H$ such that
  $d_H(v)=f(v)$ for every vertex $v\in V(G)$. The subgraph \(H\)
  is called a \emph{connected \(f\)-factor} if, in addition, $H$ is connected. A classical result of Tutte (1954) is the polynomial time algorithm to check
  whether a given graph has a specified $f$-factor. However, checking for the presence of a \emph{connected} $f$-factor is easily seen to generalize \textsc{Hamiltonian Cycle} and hence is \NP-complete. In fact, the \textsc{Connected $f$-Factor} problem remains \NP-complete even when we restrict $f(v)$ to be  at least $n^{\epsilon}$ for each vertex $v$ and $\epsilon<1$; on the other side of the spectrum of nontrivial lower bounds on $f$, the problem is known to be polynomial time solvable when $f(v)$ is at least $\frac n3$ for every vertex $v$. 

In this paper, we extend this line of work and obtain new complexity results based on restrictions on the function $f$.
In particular, we show that when $f(v)$ is restricted to be at least $\frac{n}{(\log n)^c}$,  the problem can be solved in quasi-polynomial time in general and in randomized polynomial time if $c\leq 1$. Furthermore, we show that when $c>1$, the problem is \NP-intermediate.
\end{abstract}

\section{Introduction}
The concept of $f$-factors is fundamental in graph theory, dating back to the 19th century, specifically to the work of Petersen~\cite{JP1891}. In modern terminology, an $f$-factor is defined as a spanning subgraph which satisfies degree constraints (given in terms of the degree function $f$) placed on each vertex of the graph~\cite{DW00}. Some of the most fundamental results on $f$-factors were obtained by Tutte, who gave sufficient and necessary conditions for the existence of $f$-factors~\cite{WT52}. In addition, he developed a method for reducing the $f$-factor computation problem to the  perfect matching~\cite{WT54} problem, which gives a straightforward polynomial time algorithm for the problem of deciding the existence of an $f$-factor. There are also several detailed surveys on $f$-factors of graphs, for instance by Chung and Graham~\cite{CG81}, Akiyama and Kano~\cite{JK85}, Lov{\'a}sz and Plummer~\cite{LP86}. 

Aside from work on general $f$-factors, substantial attention has been devoted to the variant of $f$-factors where we require the subgraph to be connected (see for instance the survey articles by Kouider and Vestergaard~\cite{KV05} and Plummer~\cite{PM07}). Unlike the general $f$-factor problem, deciding the existence of a connected $f$-factor is \NP-complete~\cite{GJ79,CC90}. It is easy to see that the connected $f$-factor problem (\textsc{Connected $f$-Factor}) generalizes \textsc{Hamiltonian Cycle} (set $f(v)=2$ for every vertex $v$), and even the existence of a deterministic single-exponential (in the number of vertices) algorithm is open for the problem~\cite{PhilipR14}.

The \NP-completeness of this problem has motivated several authors to study the  {\sc Connected $f$-Factor} for various restrictions on the function $f$. Cornelissen et al.~\cite{BN13} showed that \textsc{Connected $f$-Factor} remains \NP-complete even when $f(v)$ is at least $n^{\epsilon}$ for each vertex $v$ and any constant $\epsilon$ between 0 and 1. Similarly, it has been shown that the problem is polynomial time solvable when $f(v)$ is at least $\frac n3$~\cite{NR15} for every vertex $v$. Aside from these two fairly extreme cases, the complexity landscape of \textsc{Connected $f$-Factor} based on lower bounds on the function $f$, has largely been left uncharted. 

\begin{table}[t]
\begin{center}
\begin{tabular}{|l|cc|}
\hline
  $f(v) \geq $ & Complexity Class & \\
  \hline
  $n^\epsilon, \forall \epsilon>0$ & \NPC~\cite{BN13} &\\ \hline
  $\frac{n}{\polylog{n}}$ & \QP~ (Theorem \ref{thm:polylog})&\\ \hline
  $\frac{n}{\log n}$ & RP~ (Theorem \ref{thm:log})&\\ \hline
  $\frac{n}{c}, \forall c\geq 3$ & P~~(Theorem \ref{thm:constant})&\\
  \hline
\end{tabular}

	\end{center}
	
	\caption{The table depicting the known as well as new results on the complexity landscape of the  \textsc{Connected $f$-Factor} problem.}\label{tab:landscape}
\end{table}

\noindent \textbf{Our results and techniques.}
In this paper, we provide new results (both positive and negative) on solving 
\textsc{Connected $f$-Factor} based on lower bounds on the
range of $f$.
Since we study the complexity landscape of \textsc{Connected $f$-Factor} through the lens of the function $f$, it will be useful to formally capture bounds on the function $f$ via an additional ``bounding'' function $g$. To this end, we introduce the connected $g$-Bounded $f$-factor  problem (\textsc{Connected $g$-Bounded $f$-Factor}) below:

\begin{center}
\vspace{-0.2cm}
  \begin{boxedminipage}[t]{0.99\textwidth}
  \begin{quote}
  \textsc{Connected $g$-Bounded $f$-Factor}\\ \nopagebreak
  \emph{Instance}: An $n$-vertex undirected simple graph $G$ and a mapping \\$f:V(G)\rightarrow \Nat$ such that $f(v)\geq \frac n {g(n)}$.\\ \nopagebreak
  \emph{Task}: Find a connected $f$-factor $H$ of $G$.
\end{quote}
\end{boxedminipage}
\end{center}

First, we obtain a polynomial time algorithm for {\sc Connected $f$-Factor} when  $f(v)$ is at least $\frac n {c}$ for every vertex $v$ and any constant $c>1$. This result generalizes the previously known polynomial time algorithm for the case when $f(v)$ is at least $\frac n 3$. This is achieved thanks to a novel approach for the problem, which introduces a natural way of converting one $f$-factor to another by exchanging a set of edges. Here we formalize this idea using the notion of \emph{Alternating Circuits}. These allow us to focus on a simpler version of the problem, where we merely need to ensure connectedness across a coarse partition of the vertex set. Furthermore, we extend this approach to obtain a quasi-polynomial time algorithm for the {\sc Connected $f$-Factor}  when $f(v)$ is at least $\frac n {\polylog{n}}$ for every vertex. To be precise, we prove the following two theorems (see Section \ref{sec:prelim} for an explanation of the function $g$ in formal statements).


\begin{theorem}\label{thm:constant}
  For every function $g(n)\in\bigoh(1)$, \textsc{Connected $g$-Bounded $f$-Factor} can be solved in polynomial time.
\end{theorem}
   
\begin{theorem}\label{thm:polylog}
  For every $c>0$ and function $g(n)\in\bigoh((\log n)^c)$, \textsc{Connected $g$-Bounded $f$-Factor} can be solved in  time $n^{(\log{n})^{\alpha (c)}}$ where $\alpha(c)\in \bigoh(1)$.
\end{theorem}

Second, we build upon these new techniques to obtain a \emph{randomized polynomial time} algorithm which solves \textsc{Connected $f$-Factor} in the more general case where $f(v)$ is lower-bounded by $\frac n {g(n)}$ for every vertex $v$ and $g(n)\in \bigoh({\log n})$. For this, we also require  algebraic techniques that have found several applications in the design of fixed-parameter and exact algorithms for similar problems \cite{CyganNPPRW11,Wahlstrom13,GutinWY13,PhilipR14}.
Precisely, we prove the following theorem.

\begin{theorem}\label{thm:log}
  For every function $g(n)\in\bigoh(\log n)$, \textsc{Connected $g$-Bounded $f$-Factor} can be solved in polynomial time with constant error probability.
\end{theorem}

   We remark that the randomized algorithm in the above theorem has one-sided error with `Yes' answers always being correct.
%
Finally, we obtain 
a lower bound result for \textsc{Connected $f$-Factor} when $f(n)$ is at least $\frac n {(\log n)^{c}}$ for $c>1$. Specifically, in this case we show that the problem is in fact \NP-intermediate, assuming the Exponential Time Hypothesis \cite{RF01} holds. Formally speaking, we prove the following theorem.

\begin{restatable}{theorem}{lowerbound}
  \label{thm:lowerbound}
  For every $c>1$ and for every $g(n)\in\Theta((\log n)^c)$, \textsc{Connected $g$-Bounded $f$-Factor} is neither in \P\ nor \NP-hard unless the Exponential Time Hypothesis fails.
\end{restatable}

We detail the known as well as new  results on the complexity landscape of \textsc{Connected $f$-Factor} in
Table~\ref{tab:landscape}.

\noindent \textbf{Organization of the paper.}
After presenting required definitions and preliminaries in Section~\ref{sec:prelim}, we proceed to the key technique and framework used for our algorithmic results, which forms the main part of Section~\ref{sec:overview}. In Section~\ref{sec:det-alg}, we obtain both of our deterministic algorithms, which are formally given as Theorem~\ref{thm:constant} (for the polynomial time algorithm) and Theorem~\ref{thm:polylog} (for the quasi-polynomial time algorithm). 
Section~\ref{sec:random} then concentrates on our randomized polynomial time algorithm, presented in Theorem~\ref{thm:log}. 
Finally, Section~\ref{sec:npi} focuses on ruling out (under established complexity assumptions) both \NP-completeness and inclusion in \P\ possibilities for \textsc{Connected $g$-Bounded $f$-Factor} for all polylogarithmic functions $g$ in $\Theta((\log n)^c)$. 

\section{Preliminaries}\label{sec:prelim}
\subsection{Basic Definitions}

We use standard definitions and notations from West \cite{DW00}. The notation 
$d_G(v)$ denotes the \textbf{degree} of a vertex $v$ in a graph $G$. Similarly, $N_{G}(v)$ represents the set of vertices adjacent to $v$  in $G$. A \textbf{component} in a graph is a maximal subgraph that is connected. Note that the set of components in a graph uniquely determines a partition of the vertex set. A \textbf{circuit} in a graph is a cyclic sequence $v_0,e_1,v_1,\cdots ,e_k,v_k=v_0$ where each $e_i$ is of the form $\{v_{i-1},v_i\}$ and occurs at most once in the sequence. An \textbf{Eulerian circuit} in a graph is a circuit in which each edge in the graph appears. Any graph having an Eulerian circuit is called an {\em Eulerian} graph. 
 
Let $V'$ be a subset of the vertices in the graph $G$. The \textbf{vertex induced subgraph} $G[V']$ is the graph over vertex set $V'$ containing all the edges in $G$ whose endpoints are both in $V'$.  Given $E'\subseteq E(G)$, $G[E']$ is the \textbf{edge induced subgraph} of $G$ whose edge set is $E'$ and  vertex set is the set of all vertices incident on edges in $E'$. 

Given two subgraphs $G_1$ and $G_2$ of $G$, the graph $G_1\bigtriangleup G_2$ is the subgraph $G[E(G_1)\bigtriangleup E(G_2)]$. The union of the graphs $G_1,G_2,\ldots,G_r$ is the graph $\bigcup_i^r G_i$ whose vertex set is $\bigcup_i^r V(G_i)$  and edge set is $\bigcup_i^r E(G_i)$.

Given a partition $\mathcal Q=\{Q_1,Q_2,\dots,Q_r\}$ of the vertex set of $G$, the quotient graph $G/\mathcal Q$ is constructed as follows: The vertex set of $G/\mathcal Q$ is $\mathcal Q$. Corresponding to each edge $(u,v) $ in $ G$ where $u $ in $ Q_i$, $v$  in $ Q_j$, $i \neq j$, add an edge $(Q_i,Q_j)$ to $G/\mathcal Q$. Thus,  $G/\mathcal Q$ is a multigraph without loops. For a subgraph $G'$ of $G$, we say $G'$ \textbf{connects} a partition $\mathcal Q$ if $G'/\mathcal Q$ is connected. Further, we address the graph $G'$ to be a \textbf{partition connector}. A \textbf{refinement} $\mathcal Q'$ of a  partition $\mathcal Q$ is a partition of $V(G)$ where each part $Q'$ in $\mathcal Q'$ is a subset of some part $Q$ in $\mathcal Q$. This notion of partition refinement was used, e.g., by Kaiser \cite{TK12}. A {\em spanning tree of the quotient graph $G/\mathcal Q$} refers to a subgraph $T$ of  $G$ with $|\mathcal Q|$-1 edges that connects $\mathcal Q$. 
The following lemma will later be used in the analysis of the error probability of our randomized algorithm.
\begin{lemma}
  \label{lem:prob-bound}
  The following holds for every $n,c \in \mathbb{N}$ with $n>c$\emph{:}
  \begin{equation*}
    1-\frac{(c\ceil{\log n})^2}{n^2} \leq \left ( 1-\frac{1}{n^2} \right)^{c\log n} 
  \end{equation*}
\end{lemma}
\begin{proof}
Using simple term manipulations, we obtain:
  \begin{equation}\label{eqn:prob-bound-1}
    \left ( 1-\frac{1}{n^2} \right)^{c\log n} = \left ( \frac{n^2-1}{n^2} \right)^{c\log n} =  \frac{(n^2-1)^{c\log n}}{(n^2)^{c\log n}}
  \end{equation}
  Since $\frac{n^2-1}{n^2}<1$, it follows that:
  \begin{equation}
  \label{eqn:int}
\frac{(n^2-1)^{c\lceil \log n \rceil}}{(n^2)^{c\lceil \log n \rceil}} \leq   \frac{(n^2-1)^{c\log n}}{(n^2)^{c\log n}} 
  \end{equation}
By using the binomial formula, we obtain:
  \begin{eqnarray}
    (n^2-1)^{c\ceil{\log n}} & = & \sum_{i=0}^{c\ceil{\log n}}\binom{c\ceil{\log n}}{i}(n^2)^{c\ceil{\log n}-i}(-1)^i \nonumber \\
    & = & n^{2c\ceil{\log n}}+\sum_{i=1}^{c\ceil{\log n}}\binom{c\ceil{\log n}}{i}(n^2)^{c\ceil{\log n}-i}(-1)^i \nonumber \\
    & = & n^{2c\ceil{\log n}}-\sum_{i=1}^{c\ceil{\log n}}-\binom{c\ceil{\log n}}{i}(n^2)^{c\ceil{\log n}-i}(-1)^i \nonumber
    \label{eqn:prob-bound-2}
  \end{eqnarray}
  To obtain an upper bound on $\sum_{i=1}^{c\ceil{\log n}}-\binom{c \ceil{\log
    n}}{i}(n^2)^{c\ceil{\log n}-i}(-1)^i$, we show next that the absolute values
  of the terms in the sum are decreasing with increasing $i$.
  \begin{eqnarray*}
    |\binom{c\ceil{\log n}}{i}(n^2)^{c\ceil{\log n}-i}(-1)^i| & = & \binom{c\ceil{\log n}}{i}(n^2)^{c\ceil{\log n}-i}\\
    & \geq & \frac{c\ceil{\log n}}{n^2}\binom{c\ceil{\log n}}{i}(n^2)^{c\ceil{\log n}-i}\\
    & \geq & c\ceil{\log n}\binom{c\ceil{\log n}}{i}(n^2)^{c\ceil{\log n}-(i+1)}\\
        & \geq & \binom{c\ceil{\log n}}{i+1}(n^2)^{c\ceil{\log n}-(i+1)}
  \end{eqnarray*}
  The first inequality above holds because $n>c$.
  Hence, we obtain.
      \begin{eqnarray*}
    \sum_{i=1}^{c\ceil{\log n}}-\binom{c\ceil{\log n}}{i}(n^2)^{c\ceil{\log n}-i}(-1)^i & \leq &
    \sum_{i=1}^{c\ceil{\log n}}\binom{c\ceil{\log n}}{i}(n^2)^{c\ceil{\log n}-i} \\
     & \leq & (c \ceil{\log n})\cdot (c \ceil{\log n}) (n^2)^{c\ceil{\log n}-1} \\    
     & = & (c \ceil{\log n})^2 (n^2)^{c\ceil{\log n}-1}
  \end{eqnarray*}
    Putting the above back into Equation~\ref{eqn:prob-bound-1} and
  using Inequality~\ref{eqn:int} together with the expressions above, we obtain:
  \begin{eqnarray*}\label{eqn:prob-bound-2}
    \left ( 1-\frac{1}{n^2} \right)^{c\log n} & \geq & 
      \frac{(n^2-1)^{c\ceil{\log n}}}{(n^2)^{c\ceil{\log n}}} \\
    & = & \frac{n^{2c\ceil{\log n}}-\sum_{i=1}^{c\ceil{\log n}}-\binom{c\ceil{\log
          n}}{i}(n^2)^{c\ceil{\log n}-i}(-1)^i}{(n^2)^{c\ceil{\log n}}} \\
    & \geq & \frac{n^{2c\ceil{\log n}}-(c \ceil{\log n})^2 (n^2)^{c\ceil{\log n}-1}}{(n^2)^{c\ceil{\log n}}} \\
    & \geq & 1-\frac{(c\ceil{\log n})^2}{n^2}
  \end{eqnarray*}
The above concludes the proof of the lemma.
\qed \end{proof}

The function $g$ we deal with is always a positive real-valued function, defined on the set of positive integers.
For the cases we consider, the function always takes a value greater than 1. Unless otherwise mentioned,  $g(n)$ is in $\bigoh(\polylog{n})$.  Whenever $g$ is part of the problem definition, the target set of the function $f$ is the set of integers  $\{\lceil n/g(n) \rceil,\ldots,n-1 \}$. Consequently, we have the following fact.
\begin{Fact}
Let $G$ be a graph and let $f(v)\geq n/g(n)$ for each $v$ in $V(G)$. If $H$ is an $f$-factor of $G$,  then the number of components in $H$ is at most $\ceil{g(n)}-1$.
\label{fact:component-count}
\end{Fact}


\subsection{Colored Graphs, (Minimal) Alternating Circuits, and $f$-Factors}


A graph $G$ is colored if each edge in $G$ is assigned a color from the set $\{red,blue\}$.
 In a colored graph $G$, we use $R$ and $B$ to denote spanning subgraphs of $G$ whose edge sets are the set of red edges ($E(R)$) and blue edges ($E(B)$) respectively. We use this coloring in our algorithm to distinguish between edge sets of two distinct $f$-factors of the same graph $G$. A crucial computational step in our algorithms is to consider the symmetric difference between edge sets of two distinct $f$-factors and perform a sequence of edge exchanges preserving the degree of each vertex. The following definition is  used extensively in our algorithms.
 \begin{definition}
   A colored graph $A$ is  an \textbf{alternating circuit} if there
   exists an Eulerian circuit in $A$, where each pair of consecutive
   edges are of different colors.
\end{definition}

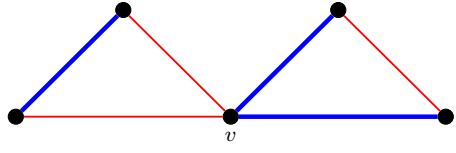
\begin{figure}
  \centering
  \begin{tikzpicture}
      \tikzstyle{every node}=[]
      \tikzstyle{gn}=[circle,fill=black,inner sep=2pt,draw=black, node distance=2cm]
      \tikzstyle{every edge}=[draw, line width=2pt]
      \tikzstyle{re}=[draw, line width=0.7pt, red]
      \tikzstyle{be}=[draw, line width=1.7pt, blue]
      \draw
      node[gn] (c11) {}
      node[gn, below left of=c11] (c12) {}
      node[gn, below right of=c11, label=below:$v$] (v) {}

      node[gn, above right of=v] (c21) {}
      node[gn, below right of=c21] (c22) {}
      ;
      \draw
      (c11) edge[re] (v)
      (v) edge[re] (c12)
      (c11) edge[be] (c12)

      (c21) edge[be] (v)
      (v) edge[be] (c22)
      (c21) edge[re] (c22)
      ;
    \end{tikzpicture}
    \label{fig:alt-cir}
    \caption{The alternating circuit $A_{11}$ illustrating that not every
      alternating circuit can be decomposed into edge-disjoint
      alternating circuits that are cycles. Here the blue edges are marked with thick lines and the red
      edges are thin.}
\end{figure}

Clearly, an alternating circuit has an even number of edges and is
connected. Further, $d_{R}(v)=d_{B}(v)$ for each $v$ in $A$. A
\textbf{minimal alternating circuit} $A$ is an alternating circuit
where each vertex $v$ in $A$ has at most two  red edges incident to
$v$. Note that alternating circuits, as opposed to Eulerian circuits,
cannot always be decomposed into edge-disjoint alternating circuits
that are cycles. As
an example, consider for each $r_1,r_2\geq 1$ the alternating circuit $A_{r_1r_2}$ which 
consists of two (edge-disjoint) cycles of length $2r_1+1$ and length $2r_2+1$, respectively, that share
one common vertex $v$. Let the coloring of the edges of $A_{r_1r_2}$ be as 
illustrated in Figure~\ref{fig:alt-cir}. Informally, the edges of
both cycles are colored in an alternating manner along each cycle so
that the edges of the first cycle incident to $v$ have the same color,
which is distinct from the color given to the edges incident with $v$
in the second cycle.
Every alternating circuit in $A_{r_1r_2}$ contains all edges of $A_{r_1r_2}$
and cannot be decomposed further into smaller alternating circuits.


\begin{Fact}\label{fac:enforcedg}
 Let $S$ be a subset of $E(G)$. An $f$-factor $H$ of $G$ containing all the edges in
 $S$, if one exists, can be computed in polynomial time. 
\end{Fact}
The fact follows from the observation that a candidate for $H\setminus S$ can be computed from an
$f'$-factor  $H'$ of the spanning subgraph $G\setminus S$, where $f'(v)=f(v)-d_{G[S]}(v)$ for each $v\in V(G)$.
Note that given a
partition $\mathcal Q=\{X,V\setminus X\}$ of $V(G)$, one can check for
the existence of an $f$-factor connecting $\mathcal Q$ in polynomial
time by iterating over each edge $e$ in the cut $[X,V\setminus X]_G$
and applying Fact \ref{fac:enforcedg} by setting $S=\{e\}$. Further,
this can be extended for any arbitrary partition $\mathcal Q$ of
constant size, or when we are provided with a spanning tree $S$ of
$G/\mathcal Q$ that is guaranteed to contain in some $f$-factor of
$G$.
\begin{definition}
Let $M$ and $H$ be two subgraphs of $G$ where each component in $M$ is Eulerian. Let $\textsf c: E(M)\rightarrow \{red,blue\}$ be the  unique  coloring function which colors the edges in $E(M)\cap E(H)$ with color red and those in $E(M)\setminus E(H)$ with color blue.  The subgraph $M$ is called a \textbf{switch} 
on $H$ if every component of the colored graph obtained by applying $\textsf c$ on $M$, is an alternating circuit.
\end{definition}

\begin{definition}
For a subgraph $M$ which is a switch on another subgraph $H$ of $G$,   we define \textbf{Switching($H$,$M$)} to be the subgraph $M\bigtriangleup H$ of $G$.
\end{definition}
We use switching as an operator where the role of the second operand is to bring in specific edges to the first, retaining the degrees of vertices by omitting some less significant edges. One can easily infer that if the result of applying the coloring function $\textsf c$ to $M$ is a minimal alternating circuit, then the switching operation replaces at most two edges incident on each vertex $v$ in $H$. 

\begin{Fact}
  Let $A$ be an alternating circuit and $S$ be a subset of edges in $A$.
  There is a polynomial time algorithm that outputs a set $\mathcal M$ of edge
  disjoint minimal alternating circuits in $A$, each of which has at
  least one edge from $S$ and such that every edge in $S$ is contained
  in some minimal alternating circuit in $\mathcal M$.
  \label{fac:min-ac-set}
\end{Fact}
It is not difficult to see the proof of Fact \ref{fac:min-ac-set}. A
skeptical reader can refer ~\cite[Lemma 6]{NR15}. Note that given $S$
and $A$, $M$ is not unique. 

\begin{Fact}
Let $H$ be an $f$-factor of $G$ and let $\mathcal{Q}$ be a partitioning of the vertex set of $G$. If $H/\mathcal Q$ is connected and $H[Q]$ is connected for each $Q$ in $\mathcal Q$, then $H$ is a connected $f$-factor.
\label{fac:connect-parts}
\end{Fact}
Fact \ref{fac:connect-parts} implies that if $H$ is not a connected $f$-factor and $H/\mathcal Q$ is connected then there exists some $Q\in \mathcal Q$ such that $H[Q]$ is not connected.  

\section{A Generic Algorithm for Finding Connected $g$-Bounded $f$-Factors}
\label{sec:overview}

Our goal in this section is to present a generic algorithm for
\textsc{Connected $g$-Bounded $f$-Factor}. In particular, we in a certain sense reduce the question of solving \textsc{Connected $g$-Bounded $f$-Factor} to solving a related problem which we call \textsc{Partition Connector}. This can be viewed as a relaxed version of the original problem, since instead of a connected $f$-factor it merely asks for an $f$-factor which connects a specified partitioning of the vertex set. A formal definition is provided below.

\begin{center}
\vspace{-0.2cm}
  \begin{boxedminipage}[t]{0.99\textwidth}
  \begin{quote}
  \textsc{Partition Connector}\\ \nopagebreak
  \emph{Instance}: An $n$-vertex graph $G$, $f:V(G)\rightarrow \mathbb{N}$, and a partition $\mathcal Q$ of $V(G)$.\\ \nopagebreak
  \emph{Task}: Find an $f$-factor of $G$ that connects $\mathcal Q$.
\end{quote}
\end{boxedminipage}
\end{center}

The algorithms for solving \textsc{Partition Connector} are presented in the later parts of this article.  Specifically, a deterministic algorithm that runs in
quasi-polynomial time whenever $g(n) = \bigoh(\polylog{n})$ (Section~\ref{sec:det-alg}) and a
randomized polynomial time algorithm for the case when
$g(n)= \bigoh(\log n)$ (Section~\ref{sec:rand-alg}).

The majority of this section is devoted to proving the key Theorem~\ref{thm:mainalg} stated below, which establishes the link between \textsc{Partition Connector} and \textsc{Connected $g$-Bounded $f$-Factor}.

  \begin{theorem}\label{thm:mainalg}
  \begin{itemize}
  \item[$(a)$] Let $g(n) \in \bigoh(\polylog{n})$. If there is a deterministic algorithm running in time
    $\bigoh^{*}(n^{2(|\mathcal{Q}|-1)})$\footnote{We use
      $\bigoh^{*}(f(n))$ to denote $\bigoh(f(n)\cdot n^{\bigoh(1)})$,
      i.e., $\bigoh^*$ omits polynomial factors, for any function $f(n)$.} for \textsc{Partition Connector},
    then there is a deterministic quasi-polynomial time algorithm for
    {\cgfa} with running time $\bigoh^{*}(n^{2g(n)})$.
  \item[$(b)$] Let $g(n)\in \bigoh(\log n)$. If there exists a randomized algorithm running in time
    $\bigoh^{*}(2^{|\mathcal{Q}|})$ with error probability
    $\bigoh(|\mathcal{Q}|^2/n^2)$ for \textsc{Partition Connector}, then there
    exists a randomized polynomial time algorithm for \cgfa{} that
     has a
    constant error probability. 
  \end{itemize}
\end{theorem}

\subsection{A generic algorithm for \textsc{Connected $g$-Bounded $f$-Factor}}

The starting point of our generic algorithm is the following observation. 
\begin{observation}\label{thm:characterization}
  \label{obs:characterization}
  Let $G$ be an undirected graph  and $f$ be a function $f:V(G)\rightarrow \mathbb N$. The graph $G$ has a connected $f$-factor
  if and only if for each partition $\mathcal Q$ of the vertex set
  $V(G)$, there exists an $f$-factor $H$ of $G$ that connects $\mathcal
  Q$. 
\end{observation}

We remark that for the running time analysis for our generic algorithm we assume that we are only dealing with 
instances of \textsc{Connected $g$-Bounded $f$-Factor}, where the number of vertices exceeds $6g(n)^4$. 
As $g(n)$ is in $\bigoh(\polylog{n})$, this
does not reduce the applicability of our algorithms, since there is a
constant $n_0$ such that $n\geq 6g(n)^4$ for every $n\geq n_0$; 
because $g(n)$ is part of the problem description, $n_0$ does not
depend on the input instance. Consequently, we can solve instances of
\textsc{Connected $g$-Bounded $f$-Factor} where $n < n_0$ by brute-force in constant time. We will
therefore assume without loss of generality in the following that $n \geq n_0$ and hence $n\geq
6g(n)^4$. 

Our algorithm constructs a  sequence  
$(H_0,\mathcal Q_0),\dotsc,(H_k,\mathcal Q_k)$ of pairs  which is {\em maximal} (cannot be extended further) satisfying the 
following properties:
\begin{enumerate}
\item[(M1)] Each $\mathcal Q_i, 0 \leq i \leq k $ is a partition of the vertex set $V(G)$, and $\mathcal Q_0 =\{V(G)\}$.
\item[(M2)]  Each $H_i, 0 \leq i \leq k$ is an $f$-factor of $G$, and $H_i$ connects $\mathcal Q_i$. 
\item[(M3)] For each $ 1 \leq i \leq k$, $\mathcal Q_i$ is a
  refinement of $\mathcal Q_{i-1}$ satisfying the following: 
  \begin{enumerate}
  \item[a)] Each part $Y$ in $\mathcal Q_i$ induces a component $H_{i-1}[Y]$ in $H_{i-1}[Q]$, for some $Q$ in $\mathcal Q_{i-1}$. 
  \item[b)] $\mathcal Q_i \neq \mathcal Q_{i-1}$.
  \end{enumerate}
\end{enumerate}

The following lemma links the existence of a connected $f$-factor to
the properties of maximal sequences satisfying (M1)--(M3).
\begin{lemma}
\label{lem:correctness}
  Let $(G,f)$ be an instance of \textsc{Connected $g$-Bounded $f$-Factor} and 
  let $\mathcal S=(H_0,\mathcal{Q}_0),\dotsc,(H_k,\mathcal{Q}_k)$ be a maximal sequence satisfying
  (M1)--(M3). Then, $G$ has a
  connected $f$-factor if and only if $H_k$ is a connected $f$-factor of $G$.
\end{lemma}
\begin{proof}
  Towards showing the forward direction of the claim,
  suppose for a contradiction that $H_k$ is not a connected $f$-factor
  of $G$. Because $H_k$ connects $\mathcal{Q}_k$ (Property (M2)), it follows from Fact \ref{fac:connect-parts} that there is some part $Q \in \mathcal{Q}_k$ such that
  $H_k[Q]$ is not connected. Consider the refinement $\mathcal{Q}_{k+1}$ of
  $\mathcal{Q}_k$ ($\mathcal Q_{k+1}\neq \mathcal Q_k$) that splits every part $Q$ in $\mathcal{Q}_{k}$ into the parts
  corresponding to the components of $H_k[Q]$. Further, because $G$ has a connected $f$-factor and Observation \ref{obs:characterization}, we obtain that there exists a connected \fa $H_{k+1}$ that connects any partition $\mathcal Q_{k+1}$. Now, the sequence $\mathcal{S}$
  could be extended by appending the pair $(H_{k+1},\mathcal Q_{k+1})$ to its
  end, a contradiction to our assumption that $\mathcal{S}$ was a
  maximal sequence.
  The reverse direction is trivial.
\qed \end{proof}

We deploy an algorithm that  incrementally computes a maximal sequence $\mathcal S$ satisfying (M1)--(M3) and thereby use the  above lemma to solve the connected $f$-factor problem by testing whether the last $f$-factor in the sequence is connected. This   involves computing $\mathcal Q_{i+1}$ from $H_i$ and $\mathcal Q_i$ followed by the computation of $H_{i+1}$ connecting $\mathcal Q_{i+1}$. However, if the number of parts in the last partition $\mathcal{Q}_k$ is allowed to grow to $n$, then such an
algorithm would eventually have to solve the connected $f$-factor problem to compute an $H_k$ satisfying (M2). Our algorithm establishes a lower bound on the size of any part $Q\in \mathcal Q_i$ and hence an upper bound on the number of parts in any partition $\mathcal{Q}_i$ in  $\mathcal S$ which in turn bounds the length of the sequence $\mathcal S$. 

The following lemma shows that given the recently computed pair
$(H,\mathcal Q)=(H_i,\mathcal Q_i)$ in the sequence, the partition
$\mathcal Q'=\mathcal Q_{i+1}$ and a candidate $H''$ for $H_{i+1}$,
one can compute a {\em better} candidate $H'$ for $H_{i+1}$ which is
{\em closer to} $H_i$ in the sense that {\em most of} the neighbors
of a vertex $v$ in $H_i$ are retained as it is, in $H'$. The properties of
$H'$ then allow us to lower-bound the size of each part  in $\mathcal Q_{i+2}$ as a function of the size of the smallest part in $\mathcal Q_{i+1}$.



\begin{lemma}
\label{lem:minimal-circuits}
  Let $(H,\mathcal{Q})$, $(H'',\mathcal{Q}')$ be two consecutive pairs 
  occurring in a sequence satisfying properties (M1)--(M3). 
  Then, there is an
  $f$-factor $H'$ of $G$ connecting $\mathcal{Q}'$ such that 
  $|N_{H'}(v) \cap Q'|\geq |N_{H}(v)\cap Q'|-2(|\mathcal{Q}'|-1)$ for
  every $Q' \in \mathcal{Q}'$ and $v \in Q'$.
  Moreover, $H'$ can be
  computed from $\mathcal{Q}'$, $H$, and $H''$ in
  polynomial time.
\end{lemma}
\begin{proof}
 
  From the premise
  that $H''$ is an \fa connecting $\mathcal Q'$, we know that there
  exists a spanning tree $T$ of $H''/\mathcal Q'$. Color the edges in
  $H$ with color red and those in $H''$ with color blue. Let $\mathcal
  A$ be the graph $H\bigtriangleup H''$. Notice that each
  component in $\mathcal A$ is an alternating circuit. Furthermore, note that the
  set $S=E(T\setminus H)$ of blue edges is a subset of $\mathcal A$ as
  $E(T)$ is a subset of $E(H'')$. Let $S_i$ be the set $A_i\cap S$
  where $A_i$ is the $i^{th}$ component in $\mathcal A$. We compute
  the set $\mathcal M_i$ of edge disjoint minimal alternating circuits using Fact \ref{fac:min-ac-set} for each $(A_i,S_i)$
  pair. The size of the set $\mathcal M_i$ is at most $|S_i|$ and hence at most $|S|$ \macs in $\mathcal M= \bigcup_i{\mathcal M_i}$.
   Let $\mathcal M_S=\bigcup_{M\in \mathcal M}M$ and $H'$ be the \fa defined as  Switching($H$,$\mathcal
  M_S$).  We
  argue that this switching operation removes at most
  $2(|\mathcal{Q}'|-1)$ edges 
  incident on any vertex $v$ in $H[Q']$ for every $Q'\in
  \mathcal{Q}'$. 
  
   Considering the fact that the \macs in $\mathcal M$ are edge disjoint, we visualize switching with $\mathcal M_S$ as a sequence of switching operations on $H$ each with a distinct \mac $M$ in $\mathcal M$. In each such $M$, the number of red edges incident on a vertex $v$ that leaves $H$ during
  switching is at most  two and the operation Switching($H$,$\mathcal M_S$) retains at least $N_H(v)-2|\mathcal M|$ neighbors of each vertex. Thus, for any subset $Q'$ of
  $V(G)$ if we consider the subgraph $H[Q']$ alone, it must be the case that
  $|N_{H[Q']}(v)\cap N_{H'[Q']}(v)|\geq |N_{H[Q']}(v)|-2|\mathcal M|$
  for each $v$ in $Q'$. Furthermore, $|\mathcal{M}|$ is at most
  $|S|=|\mathcal Q'|-1$. Since the set $E(T)$ is a subset of $E(H')$,  $H'$ connects $\mathcal Q'$. From Fact  \ref{fac:min-ac-set}, the computation of $\mathcal M$
  and hence of $H'$ takes polynomial time. This completes the proof of the lemma.
\qed \end{proof}

By employing the above lemma, our algorithm ensures that the maximal
sequence $(H_0,\mathcal{Q}_0),$ $\dotsc,(H_k,\mathcal{Q}_k)$ so constructed satisfies the following additional property:
\begin{enumerate}
\item[(M4)] For every $1 \leq i \leq k$, every $Q \in \mathcal{Q}_i$
  and $v \in Q$ it
  holds that
  $|N_{H_{i}}(v)\cap Q|\geq |N_{H_{i-1}}(v)\cap Q|-2(|\mathcal{Q}_i|-1)$.
\end{enumerate}
This property plays a key role in the analysis of our algorithm 
as it allows us to bound the number of parts in each partition
$\mathcal{Q}_i$. Towards this aim we require the following auxiliary
lemma.
\begin{lemma}
\label{lem:min-degree-comp}
  Let $\mathcal{S}=(H_0,\mathcal{Q}_0),\dotsc,(H_k,\mathcal{Q}_k)$ be
  a sequence satisfying properties (M1)--(M4). Then,
  $|N_{H_i}(v) \cap Q|\geq f(v)-\sum_{1 \leq j \leq i}2(|\mathcal{Q}_j|-1)$
  for every $i$ with $1 \leq i \leq k$, $Q \in \mathcal{Q}_i$ and $v \in Q$.
\end{lemma}

\begin{proof}
  We show the claim by induction on $i$ starting from $i=1$. Let $Q
  \in \mathcal{Q}_1$ and $v \in Q$. Because $H_0$ is an $f$-factor of
  $G$ and $Q$ is a component of $H_0$, we obtain that $|N_{H_0}\cap
  Q|=f(v)$. Using Property (M4) for $i=1$, we obtain
  $|N_{H_{1}}(v)\cap Q|\geq |N_{H_{0}}(v)\cap
  Q|-2(|\mathcal{Q}_1|-1)=f(v)-2(|\mathcal{Q}_1|-1)$, as required. 
  Hence assume that the claim holds for $i-1$ and we want to show the
  claim for $i$. Let $Q_i \in \mathcal{Q}_i$ and $v \in Q_i$ and let $Q_{i-1}$
  be the part in $\mathcal{Q}_{i-1}$ containing $Q_i$. Note that $v\in Q_{i-1}$. 
  From the induction hypothesis we obtain that
  $|N_{H_{i-1}}(v) \cap Q_{i-1}|\geq f(v)-\sum_{1\leq j \leq
    i-1}2(|\mathcal{Q}_j|-1)$. Because $Q_{i}$ is a component of
    $H_{i-1}[Q_{i-1}]$, it holds that $|N_{H_{i-1}}(v) \cap
    Q_{i-1}|=|N_{H_{i-1}}(v) \cap Q_{i}|$. Hence together with Property
  (M4), we obtain 
    \begin{align*}
            |N_{H_i}(v) \cap Q_i| \geq\; &  |N_{H_{i-1}}(v)\cap
  Q_i|-2(|\mathcal{Q}_i|-1) && 
        \\
        =\; &   |N_{H_{i-1}}(v)\cap
  Q_{i-1}|-2(|\mathcal{Q}_i|-1) &&\\
        \geq \; &     f(v)-(\sum_{1\leq j \leq
    i-1}2(|\mathcal{Q}_j|-1))-2(|\mathcal{Q}_i|-1) &&\\
        =\; &  f(v)-(\sum_{1\leq j \leq
    i}2(|\mathcal{Q}_j|-1))
       \end{align*}
    as required. This completes the proof of the lemma.
\qed \end{proof}

Recall that $f(v)$ is at least $n/g(n)$ for each $v\in V(G)$. Our next step is to show that the length of the 
maximal sequence constructed by our algorithm does not exceed $g(n)+1$.
\begin{lemma}
\label{lem:bound-part-size}
  Let $\mathcal{S}=(H_0,\mathcal{Q}_0),\dotsc,(H_k,\mathcal{Q}_k)$ be
  a maximal sequence satisfying properties (M1)--(M4). Then,
  $|\mathcal{Q}_i|\leq g(n)+1$ for every $i$ with $0\leq i \leq
  k$. Moreover, the length of $\mathcal{S}$ is at most
  $g(n)+1$.
\end{lemma}

\begin{proof}
  The claim clearly holds for $\mathcal{Q}_0$. It also holds for $\mathcal{Q}_1$ because 
  the parts in $\mathcal{Q}_1$ correspond to the components of $H_0$,
  which are at most $g(n)$ due to
  Fact~\ref{fact:component-count}. Assume for a contradiction that
  the claim does not hold and let
  $\mathcal{S}=(H_0,\mathcal{Q}_0),\dotsc,(H_k,\mathcal{Q}_k)$ be a
  maximal sequence satisfying (M1)--(M4) witnessing this and let $\ell$
  be the smallest integer such that $|\mathcal{Q}_\ell|>g(n)+1$. Then,
  $\ell>1$ and $|\mathcal{Q}_{\ell-1}|\leq g(n)+1$. Because
  $|\mathcal{Q}_0|=1$ and for every $i$, $|\mathcal{Q}_i|$ is larger
  than $|\mathcal{Q}_{i-1}|$, we obtain that $i\leq |\mathcal{Q}_i|-1$
  for every $i$. Hence, $\ell-1 \leq |\mathcal{Q}_{\ell-1}|-1\leq
  g(n)+1-1=g(n)$ or in other words $\ell \leq g(n)+1$.

  From Lemma~\ref{lem:min-degree-comp}, we obtain that
  \begin{align*}
    |N_{H_{\ell-1}}(v) \cap Q| \geq\; &  f(v)-\sum_{1 \leq j <
      \ell}2(|\mathcal{Q}_j|-1) && 
    \\
    \geq\; &   \frac{n}{g(n)}-\sum_{1 \leq j <
      \ell}2(g(n)) &&\\ 
    \geq\; & \frac{n}{g(n)}-2(\ell-1) g(n) &&
    \\
    \geq \; &     \frac{n}{g(n)}-2(g(n))^2
  \end{align*}
  for every $Q \in \mathcal{Q}_{\ell-1}$ and $v \in Q$. 
  This implies that every component of $H_{\ell-1}[Q]$ for some $Q \in
  \mathcal{Q}_{\ell-1}$, and hence also every part of $\mathcal{Q}_\ell$ 
  has size at least $n/g(n)-2g(n)^2+1$.
  Since $|\mathcal{Q}_\ell|>g(n)+1$, we conclude that the number $n$ of
  vertices of $G$ is greater than $(n/g(n)-2g(n)^2+1)(g(n)+1)$.
  Rearranging for $n$ we obtain that
  $n<2g(n)^4+2g(n)^3+g(n)^2+g(n)<6g(n)^4$ which contradicts our
  assumption that $n \geq 6g(n)^4$. Since $\mathcal{Q}_i$ is a
  proper refinement of $\mathcal{Q}_{i+1}$ for every $i$ with $1 \leq
  i < k$ and $|\mathcal{Q}_0|=1$, we
  infer that the length of the sequence $\mathcal{S}$ is at most $g(n)+1$. This completes the proof of the lemma.
\qed \end{proof}

We are now ready to prove the main theorem of this section which outlines  how the running time of \textsc{Connected $g$-Bounded $f$-Factor} is dominated by the \textsc{Partition Connector} module.
\begin{proof}[Theorem~\ref{thm:mainalg}]
  We  present an algorithm for \cgfa\  that employs an algorithm for
  \textsc{Partition Connector} as a subroutine. All parts of the
  algorithm apart from the subroutine \textsc{Partition Connector}
  will be deterministic and run in polynomial time. The main idea is
  to construct a maximal sequence $\mathcal{S}=(H_0,\mathcal{Q}_0),
  \dotsc, (H_k,\mathcal{Q}_k)$ satisfying properties
  (M1)--(M4). Recall our assumption that $n\geq 6g(n)^4$. 
  Let $(G,f)$ be an instance of \cgfa. The algorithm starts by
  computing an arbitrary $f$-factor $H_0$. If no $f$-factor exists, then
  clearly the algorithm
   reports failure. If on the other hand the computed $f$-factor $H_0$
  is already connected, then the algorithm returns $H_0$ and exits.

  Observe that $(H_0,\mathcal{Q}_0)$, where $\mathcal{Q}_0=\{V(G)\}$, is a valid
  starting pair for a sequence $\mathcal S$ satisfying properties (M1)--(M4). Further, the algorithm extends the sequence $\mathcal S$ by adding successors as long as one exists. The sequence is extended by invoking a recursive subroutine
  \textbf{Restricted-$f$-factor} with parameters $(G,f)$ and the most recently added pair $(H,
  \mathcal{Q})$ to compute a
  new pair $(H',\mathcal{Q}')$ that can be appended to the sequence,  if one exists. Otherwise, the procedure  
  concludes that $\mathcal S$ can no longer be extended, in which case
  it either returns a connected $f$-factor of $G$ or reports nonexistence of one. 
  The subroutine \textbf{Restricted-$f$-factor} works as follows.
  
  The procedure starts by computing a refinement $\mathcal{Q}'$ of $\mathcal{Q}$
  containing one part $V(C)$ for every component $C$ in $H[Q]$ where $Q$ is a part in $\mathcal Q$. If $\mathcal{Q}'=\mathcal{Q}$ then because of
  Fact~\ref{fac:connect-parts}, $H$ already constitutes a connected
  $f$-factor of $G$ and the procedure correctly returns $H$. Otherwise, 
  the procedure calls the provided algorithm for \textsc{Partition
    Connector} on $G$, $f$, and $\mathcal{Q}'$ to obtain an $f$-factor
  $H''$ connecting $\mathcal{Q}'$. If the provided algorithm for \textsc{Partition
    Connector} returns failure, the procedure also returns failure, relying on Observation \ref{obs:characterization} 
 and assuming that no $f$-factor connecting $\mathcal{Q}'$ exits. Otherwise, observe that the pair
  $(H'',\mathcal{Q}')$ already constitutes a valid successor of the
  pair $(H,\mathcal{Q})$ in any sequence satisfying properties
  (M1)--(M3). To ensure Property (M4), the procedure now
  calls a polynomial time subroutine on the pairs $(H,\mathcal{Q})$ and
  $(H'',\mathcal{Q}')$ to obtain the desired $f$-factor $H'$
  connecting $\mathcal{Q}'$ and such that the pairs $(H,\mathcal{Q})$ and
  $(H',\mathcal{Q}')$ satisfy Property (M4). The existence of such a polynomial time subroutine is from   Lemma~\ref{lem:minimal-circuits}. The procedure now calls
  itself on the pair $(H',\mathcal{Q}')$. This completes the
  description of the algorithm. 

  Note that given the correctness of the
  algorithm for \textsc{Partition Connector} the correctness of
  the algorithm follows from Lemma~\ref{lem:correctness}. Let us now
  analysis the running time of the algorithm. Apart from the calls to the provided subroutine for \textsc{Partition
    Connector}, all parts of the
  algorithm  run in polynomial time. Because the algorithm calls the provided algorithm for
  \textsc{Partition Connector} at most once for every pair
  $(H,\mathcal{Q})$ in a maximal sequence satisfying properties
  (M1)--(M4), we obtain from Lemma~\ref{lem:bound-part-size} that the
  number of those calls is bounded by $g(n)+1$. Moreover, from the same
  lemma, we obtain that the size of a partition $\mathcal{Q}$ given as
  an input to the algorithm for \textsc{Partition Connector} is at
  most $g(n)+1$. Hence, if \textsc{Partition Connector} can be
  solved in time $\bigoh^{*}(n^{2(|\mathcal{Q}|-1)})$, then the algorithm runs
  in time $\bigoh^{*}(g(n)n^{2(g(n))})$ showing the first statement of the
  theorem. Similarly, if \textsc{Partition
    Connector} can be solved in time
  $\bigoh^{*}(2^{|\mathcal{Q}|})$, then the algorithm runs in time
  $\bigoh^{*}(g(n)2^{g(n)+1})$, which given $g(n) \in \bigoh(\log n)$
  shows that the algorithm claimed in the second statement of the
  theorem runs in polynomial time. We use the following lemma to prove the second part of the theorem.
\begin{lemma}
\label{thm:rand-part-conn}
  The \textsc{Partition Connector} can be solved by a
  randomized algorithm with running time $\bigoh^{*}(2^{|\mathcal{Q}|})$ and error probability
  $\bigoh(1-(1-\frac{1}{n^2})^{|\mathcal{Q}|})$.
\end{lemma}
  It remains to show that the randomized algorithm has the stated error probability.
  Towards this aim we calculate a lower bound on the
  success probability of the algorithm, i.e., the probability that the
  algorithm returns a connected $f$-factor of $G$ if such an
  $f$-factor exists. Hence, let us suppose that $G$ has a connected
  $f$-factor. It follows from Observation~\ref{obs:characterization}
  that $G$ contains an $f$-factor connecting $\mathcal{Q}$ for every
  partition $\mathcal{Q}$ of its vertex set. Hence every call to the 
  subroutine \textbf{Partition Connector} is made for a
  ``Yes''-instance, which together with Lemma~\ref{thm:rand-part-conn} implies that
  every such call succeeds with probability at least 
  $(1-\frac{1}{n^2})^{|\mathcal{Q}|}$. Because $|\mathcal{Q}|\leq
  g(n)+1 \in \bigoh(\log n)$, we obtain from Lemma~\ref{lem:prob-bound}
  that this probability is at least $(1-\frac{c\ceil{\log n}^2}{n^2})$
  for some constant $c$.
  Since
  there are at most $g(n)+1=c\log n$ such
  calls, the probability that the algorithm succeeds for all of
  these calls is hence at least $(1-\frac{c\ceil{\log
      n}^2}{n^2})^{c\log n}>0$, as
  required. This completes the proof of the theorem.
\qed \end{proof}

\subsection{A Quasi-polynomial Time Algorithm for Polylogarithmic Bounds}
\label{sec:det-alg}
In this section, we prove Theorem \ref{thm:constant} and Theorem \ref{thm:polylog}. In fact, we prove a more general result, from which both theorems directly follow.

\begin{restatable}{theorem}{genthm}\label{thm:general}
  For every $c>0$ and function $g(n)\in\bigoh((\log n)^c)$, the {\cgfa} problem can be solved in ${\bigoh}^{*}(n^{2g(n)})$ time.
\end{restatable}

We make use of the following simple lemma.
\begin{lemma}
Let $G$ be a graph having a connected $f$-factor. Let $\mathcal Q$ be  a partition of the vertex set $V(G)$. There exists a spanning tree $T$ of $G/\mathcal Q$ such that for some \fa $H$ of $G$, $E(T)\subseteq E(H)$. Furthermore, $H$ can be computed from $T$ in polynomial time.
\label{lem:connect-parts}
\end{lemma}
\begin{proof}
Let $H'$ be a connected $f$-factor of $G$. For any partition $\mathcal Q$ of the vertex set, it follows from Observation \ref{thm:characterization} that $H'/\mathcal Q$ is connected. Consider a spanning tree $T$ of $H'/\mathcal Q$.  Clearly, there exists at least one $f$-factor $H$ containing $E(T)$ and hence $H/\mathcal Q$ is connected. Once we have $E(T)$, $H$ can be computed in polynomial time using Fact \ref{fac:enforcedg}. 
 \qed \end{proof}

In light of Theorem~\ref{thm:mainalg}, it now suffices to prove the following Lemma~\ref{lem:alg-part-conn-poly}, from which Theorem~\ref{thm:general} immediately follows.

\begin{lemma}\label{lem:alg-part-conn-poly}
  \textsc{Partition Connector} can
  be solved in time $\bigoh^{*}(n^{2(|\mathcal{Q}|-1)})$.
\end{lemma}
\begin{proof}
  It follows from Lemma~\ref{lem:connect-parts} that we can solve
  \textsc{Partition Connector} by going over all spanning
  trees $T$ of $G / \mathcal{Q}$ and checking for each of them 
  whether there is an $f$-factor of $G$ containing the edges of $T$. 
  The lemma now follows because the number of spanning trees of
  $G / \mathcal{Q}$ is at most $\binom{|E(G)|}{|\mathcal{Q}|-1}$,
  which is upper bounded by $\bigoh(n^{2(|\mathcal{Q}|-1)})$, and for every
  such tree $T$ we can check the existence of an $f$-factor containing
  $T$ in polynomial time.
\qed \end{proof}

\section{A Randomized Polynomial Time Algorithm for Logarithmic Bounds}
\label{sec:rand-alg}
\label{sec:random}
 In this section we prove Theorem \ref{thm:log}. Due to
 Theorem~\ref{thm:mainalg}, it is sufficient for us to provide a
 randomized algorithm for \textsc{Partition Connector} 
 with running time $\bigoh^{*}(2^{|\mathcal{Q}|})$ and error probability
    $\bigoh(g(n)^2/n^2)$.
    This is precisely what we do in the rest of this section (Lemma \ref{thm:rand-part-conn}). As a first step, we design an algorithm for the ``existential version'' of the problem which we call $\exists$-Partition Connector and define as follows.\\

\defproblem{\textsc{$\exists$-Partition Connector}}{A graph $G$ with $n$ vertices,
  $f:V(G)\rightarrow \mathbb{N}$, and a partition $\mathcal Q$ of
  $V(G)$.}{Is there an $f$-factor of $G$ that connects $\mathcal Q$?}

We then describe how to use our algorithm for this problem as a subroutine in our algorithm to solve \textsc{Partition Connector}.

\subsection{Solving {\sc $\exists$-Partition Connector} in Randomized Polynomial Time}

The objective of this subsection is to prove the following lemma which implies a randomized polynomial time algorithm for {\sc $\exists$-Partition Connector} when $g(n)\in\bigoh(\log n)$.

\begin{lemma}
\label{lem:existfactor} There exists an algorithm that, given the graph $G$, a function $f:V(G)\to {\mathbb N}$, and a partition $\cal Q$ of $V(G)$, runs in time $\bigoh^{*}(2^{|{\cal Q}|})$ and outputs

\begin{itemize} \item {\sc NO} if $G$ has no $f$-factor connecting $\cal Q$
\item {\sc YES} with probability at least  $1-\frac{1}{n^2}$ otherwise. 	
\end{itemize}

\end{lemma}

 We design this algorithm by starting from the exact-exponential algorithm in \cite{PhilipR14} and making appropriate  modifications. During the description, we point out the main differences between our algorithm and that in \cite{PhilipR14}. We now proceed to the details of the algorithm. 
We begin by recalling a few important definitions and known results on $f$-factors. These are mostly standard and are also present in \cite{PhilipR14},  but since they are required in the description and proof of correctness of our algorithm, we state them here.

\begin{definition}[\(f\)-Blowup]\label{def:f_blowup}
  Let \(G\) be a graph and let \(f:V(G)\to\mathbb{N}\) be such
  that \(f(v)\leq{}deg(v)\) for each \(v\in{}V(G)\). Let \(H\) be the graph  constructed as follows:
  \begin{enumerate}
  \item For each vertex \(v\) of \(G\), we add a vertex set
    \(A(v)\) of size $f(v)$ to \(H\). 
  
  \item For each edge \(e=\{v,w\}\) of \(G\) we add to \(H\)
    vertices $v_e$ and $w_e$ and edges $(u,v_e)$ for every $u\in
    A(v)$ and $(w_e,u)$ for every $u\in A(w)$.  Finally, we add the edge $(v_e,w_e)$.
       \end{enumerate}
  This completes the construction. The graph \(H\) is called the 
  \emph{\(f\)-blowup} of graph \(G\).  We use ${\cal B}_{f}(G)$ to denote the
$f$-blowup of $G$. We omit the subscript when there is no scope for ambiguity.
\end{definition}

\begin{definition}[Induced $f$-blowup]
  For a subset $S\subseteq V(G)$, we define the $f$-blowup
  of $G$ \emph{induced} by $S$ as follows. Let the 
  $f$-blowup of $G$ be $H$. Begin with the graph $H$ and for
  every edge $e=(v,w)\in E(G)$ such that $v\in S$ and $w\notin S$,
  delete the vertices $v_e$ and $w_e$ from $H$.  Let the graph $H'$ be the union of those connected
  components of the resulting graph which contain the vertex sets
  \(A(v)\) for vertices \(v\in{}S\). Then, the graph $H'$ is
  called the $f$-blowup of $G$ \emph{induced} by the set
  $S$ and is denoted by ${\cal B}_f(G)[S]$.

\end{definition}

We now recall the relation between perfect matchings in the $f$-blowup and $f$-factors (see Figure \ref{fig:blowup}). 

\begin{lemma}[\cite{WT54}]
\label{lem:matching_equivalence}
  A graph \(G\) has an \(f\)-factor if and only if the \(f\)-blowup of \(G\)
 has a perfect matching.  \end{lemma}

 \begin{figure}[t]
 \begin{center}
  \includegraphics[height=220 pt, width=200 pt]{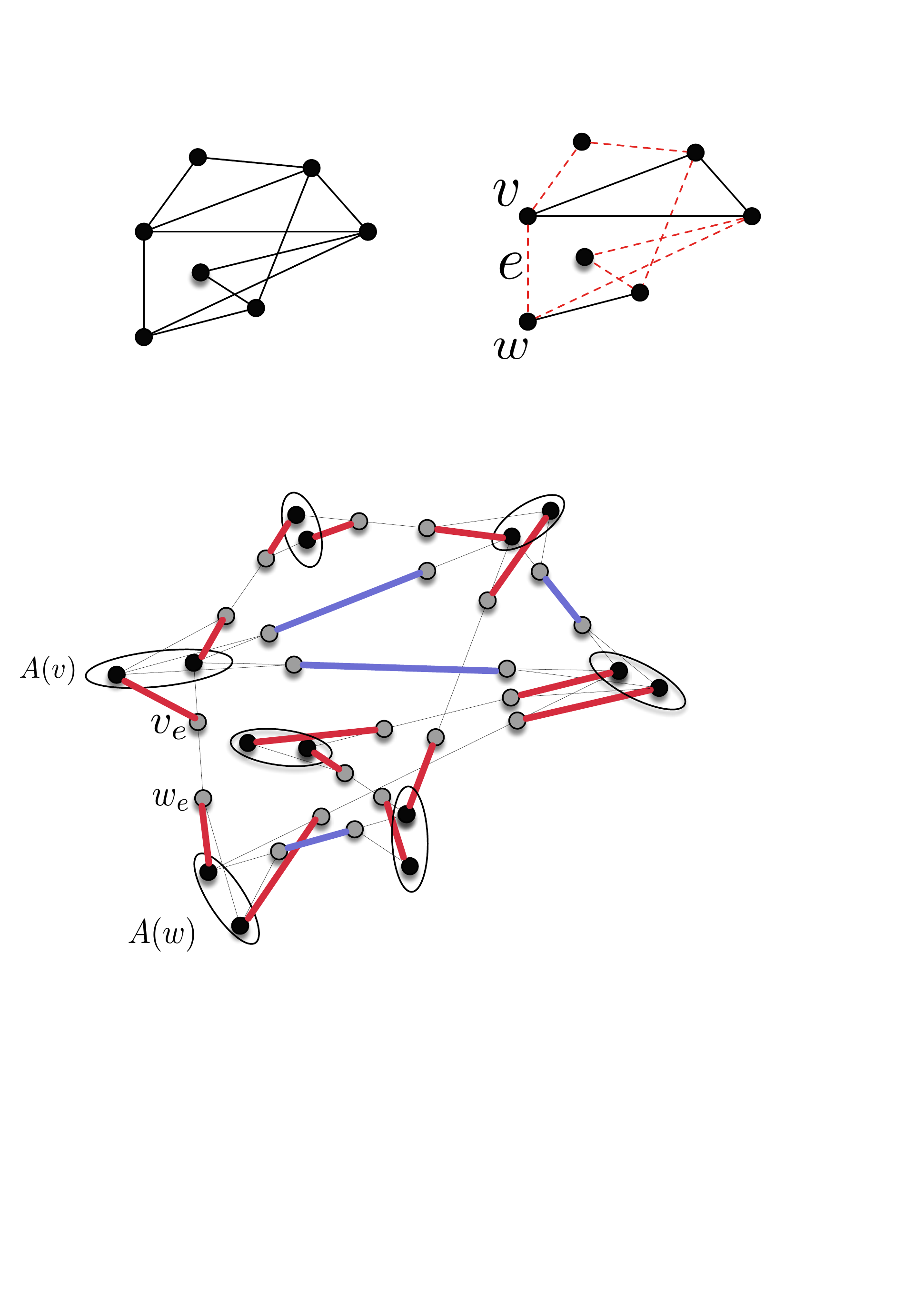}
  \caption{An illustration of a graph $G$ with a 2-factor $H$ (the red dashed edges) and one possible corresponding perfect matching in $\cB(G)$ (thick edges). It is important to note that an edge $e=(v,w)$ is \emph{not} in $H$ if and only if the edge $(v_e,w_e)$ is present in the corresponding perfect matching.}
  \label{fig:blowup}
  \end{center}
\end{figure}


The relationship between the Tutte matrix and perfect matchings is well-known and this has already been exploited in the design of fixed-parameter and exact algorithms \cite{Wahlstrom13,GutinWY13,PhilipR14}.

 \begin{definition}[Tutte matrix] The \emph{Tutte matrix} of a
  graph $G$ with $n$ vertices is an $n\times n$ skew-symmetric
  matrix $T$ over the set $\{x_{ij}|1\leq i< j\leq \vert V(G)\vert\}$ of indeterminates whose \((i,j)^{th}\) element is defined
  to be 
  
\[ 
T(i,j) = \left \{ \begin{array}{lll}
    x_{ij} & \mbox{if $\{i,j\}\in E(G)$ and $i<j$}\\
    -x_{ji} & \mbox{if $\{i,j\}\in
      E(G)$ and $i>j$}\\
        
    0 & \mbox{otherwise}\\
        
  \end{array} \right.  
\] We use \({\cal T}(G)\) to denote the Tutte matrix of the graph \(G\) 
. 

\end{definition}

 Following terminology in \cite{PhilipR14},  when we refer to expanded forms of \emph{succinct} representations
(such as summations and determinants) of polynomials, we use the
term \emph{na{i}ve expansion} (or summation) to denote that
expanded form of the polynomial which is obtained by merely
writing out the operations indicated by the succinct
representation. We use the term \emph{simplified expansion} to
denote the expanded form of the polynomial which results after we
apply all possible simplifications (such as cancellations) to a
na{i}ve expansion. We call a monomial \(m\) which has a non-zero
coefficient in a simplified expansion of a polynomial \(P\), a
\emph{surviving} monomial of \(P\) in the simplified
expansion. Let \(\det{{\cal T}(G)}\) denote the determinant of the Tutte matrix of the graph $G$. 

\begin{proposition}[\cite{WT47}]
\label{prop:perfect_matching_determinant}
%
  \(\det{{\cal T}(G)}\) is identically zero
  when expanded and simplified over a field of characteristic two
  if and only if the graph \(G\) does not have a perfect matching.
\end{proposition}

The following basic facts about the Tutte matrix \(\cT(G)\) of a
graph \(G\) are well-known. 
%
%
When evaluated over any field of
characteristic two, the determinant and the permanent of the
matrix \(\cT(G)\) (indeed, of any matrix) coincide. That is, 
\begin{equation}\label{eqn:determinant_permanent}
\det{\cT(G)}=\textup{perm}(\cT(G))=\sum_{\sigma\in{}S_{n}}\prod_{i=1}^{n}\cT(G)(i,\sigma(i)),
\end{equation}

where \(S_{n}\) is the set of all \emph{permutations} of
\([n]\). Furthermore, there is a one-to-one correspondence between
the set of all \emph{perfect matchings} of the graph \(G\) and the
\emph{surviving monomials} in the above expression for
\(\det{\cT(G)}\) when its simplified expansion is computed over any
field of characteristic two. We formally state and give a proof of the latter fact for the sake of completeness and because we intend to use this particular formulation of it.

\begin{lemma}\label{lem:perfect_matching_determinant}
	  Let ${\cT(G)}$ and $\det{\cT(G)}$ be as defined above. Then the following statements hold.
\begin{enumerate}\item  If \(M=\{(i_{1},j_{1}),(i_{2},j_{2}),\dotsc,(i_{\ell},j_{\ell})\}\)
  is a perfect matching of a graph \(G\), then the product
  \(\prod_{(i_{k},j_{k})\in{}M}x_{i_{k}{}j_{k}}^2\) appears exactly once in the naive expansion and hence as a
  surviving monomial in the sum on the right-hand side of
  \autoref{eqn:determinant_permanent} when this sum is expanded
  and simplified over any field of characteristic two. 
  \item Conversely,
  each surviving monomial in a simplified expansion of this sum
  over a field of characteristic two must be of the form
  \(\prod_{(i_{k},j_{k})\in{}M}x_{i_{k}{}j_{k}}^2\) where 
  \(M=\{(i_{1},j_{1}),(i_{2},j_{2}),\dotsc,(i_{\ell},j_{\ell})\}\) is a perfect matching 
  of \(G\).
  \end{enumerate}
\end{lemma}

\begin{proof}
For the first statement, consider the permutation $\sigma\in S_n$ comprising precisely the 2-cycles $\{(i_{1},j_{1}),(i_{2},j_{2}),\dotsc,(i_{\ell},j_{\ell})\}$. The corresponding monomial given by the definition of ${\cT(G)}$ over a field of characteristic two is precisely \(\prod_{(i_{k},j_{k})\in{}M}x_{i_{k}{}j_{k}}^2\). For every other permutation $\sigma'\in S_n$, the corresponding monomial given by the definition of ${\cT(G)}$ contains at least one variable $x_{i_rj_r}$ where $i_r$ is not mapped to $j_r$ in $\sigma$. This implies that no other monomial in the naive expansion of  \autoref{eqn:determinant_permanent} is equal to \(\prod_{(i_{k},j_{k})\in{}M}x_{i_{k}{}j_{k}}^2\) even when considered over a field of characteristic two.  This completes the argument for the first statement.

We now consider the second statement.
First of all, since we only consider simple graphs, we have that for any permutation $\sigma\in S_n$ with a fixed point, the corresponding monomial is 0 since $x_{ii}=0$ for every $i\in |V(G)|$. 
Let $S_n^{\geq 3}$ denote the set of all permutations in $S_n$ with a cycle of length at least 3.
We now argue that for any permutation $\sigma\in S_n^{\geq 3}$, the corresponding monomial vanishes in the simplified expansion of \autoref{eqn:determinant_permanent}. 
 In order to do so, we give a bijection $\beta:S_n\to S_n$ such that (a) for every $\sigma\in S_n\setminus S_n^{\geq 3}$, $\beta(\sigma)=\sigma$, (b) for every $\sigma\in S_n$, $\beta(\beta(\sigma))=\sigma$, and (c) for every $\sigma\in S_n$, the monomials corresponding to $\sigma$ and $\beta(\sigma)$ are equal over any field of characteristic two. 
 
 We first define $\beta(\sigma)$ for a $\sigma\in S_n^{\geq 3}$ as follows. Note that we have  already fixed an ordering of the vertices of $G$. Let $v$ be the first vertex of $G$ in this ordering which appears in a cycle of length at least 3 in $\sigma$ and let $C$ denote this cycle.  We now define $\beta(\sigma)$ to be the permutation obtained from $\sigma$ by inverting $C$ and leaving every other cycle unchanged. Finally, for every $\sigma\in S_n\setminus S_n^{\geq 3}$, simply set $\beta(\sigma)=\sigma$.

 It is straightforward to see that the resulting mapping $\beta$ is indeed a bijection and moreover, $\beta(\beta(\sigma))=\sigma$ for every $\sigma\in S_n$ as required. Finally, it follows from the definition of $\det{\cT(G)}$ that over a field of characteristic two, the factor of the monomial corresponding to $\sigma$ contributed by any cycle $C$ is the same as that contributed by the inverse of this cycle to the monomial corresponding to $\beta(\sigma)$. Hence we have the third property and conclude that for any permutation $\sigma\in S_n^{\geq 3}$, the corresponding monomial vanishes in the simplified expansion of \autoref{eqn:determinant_permanent}. This implies that the only surviving monomials are those corresponding to permutations in $S_n\setminus S_n^{\geq 3}$ without a fixed point, implying that these permutations comprise only 2-cycles. This in turn implies that any such surviving monomial must correspond to a perfect matching of $G$ as required. 
 This completes the proof of the lemma.
	\qed
\end{proof}

\begin{lemma}[Schwartz-Zippel Lemma, \cite{Schwartz80,Zippel79}]
Let $P(x_1,\dots,x_n)$ be a multivariate polynomial of degree at most $d$ over a field ${\mathbb F}$ such that $P$ is not identically zero. Furthermore, let $r_1,\dots, r_n$ be chosen uniformly at random from ${\mathbb F}$. Then, $$Prob[P(r_1,\dots,r_n)=0]\leq \frac{d}{\vert {\mathbb F}\vert}.$$
\end{lemma}

\begin{definition}\label{def:monomial}
For a partition of $V(G)$, $\cQ=\{Q_1,\dots, Q_\ell\}$ and a subset $I\subseteq [\ell]$, we denote by $\cQ(I)$ the set $\bigcup_{i\in I}Q_i$. Furthermore, 
  with every set $\emptyset\neq I\subset [\ell]$, we associate a specific
  monomial $m_I$ which is defined to be the product of the  terms
  $x_{ij}^2$ where $i<j$ and $\{i,j\}=\{v_e,w_e\}$,  $e=(v,w)\in E(G)$ crosses the cut $(\cQ(I),\overline{\cQ(I)})$ and  \(v_{e},w_{e},\) are as
  in Definition~\ref{def:f_blowup} of the 
  \(f\)-blowup ${\cal B}(G)$ of \(G\). For $I=[\ell]$, we define $m_I=1$.
  \end{definition}

From now on, for a set $X\subseteq V(G)$, we denote by $\overline{X}$ the set $V(G)\setminus X$. Also, since we always deal with a fixed graph $G$ and function $f$,  for the sake of notational convenience, we refer to the graph $\cB_f(G)$ simply as $\cB$.
We now define a
polynomial $P_\cQ(\bar x)$ over the indeterminates from the 
Tutte matrix ${\cal T}({\cal B})$ of the \(f\)-blowup of \(G\), as
follows:
 
  \begin{equation}\label{eqn:connectivity_polynomial_editing}
    P_\cQ(\bar{x})=\sum_{\{1\}\subseteq I\subseteq{}[\ell]\; }
    (\det{{\cal T}({\cal B}[\cQ(I)])})\cdot (\det{{\cal T}({\cal B}[\overline{\cQ(I)}])}) \cdot m_{I},
  \end{equation} 
  where if a graph \(H\) has no vertices or edges then we set
  \(\det{{\cal T}(H)}=1\).   
  In what follows, we always deal with a fixed partition $\cQ=\{Q_1,\dots, Q_\ell\}$ of $V(G)$.

  \begin{remark} The definition of the polynomial $P_\cQ(\bar x)$ is  the main difference between our algorithm and the algorithm in \cite{PhilipR14}. The rest of the details are identical.  The main algorithmic consequence of this difference is the time it takes to evaluate this polynomial at a given set of points. This is captured in the following lemma whose proof follows from the fact that determinant computation is a polynomial time solvable problem. \end{remark}

\begin{lemma}\label{lem:evaluation_fast}
	Given values for the variables $x_{ij}$ in matrix $\cT(\cB)$, the polynomial $P_\cQ(\bar x)$ can be evaluated over a field $\mathbb F$ of character 2 and size $\Omega(n^6)$ in time $\bigoh^{*}(2^\ell)$.
\end{lemma}

\begin{proof}
The algorithm to evaluate $P_\cQ(\bar x)$ over the field $\mathbb F$ proceeds as follows. Given the values for the variables $x_{ij}$ in the matrix $\cT(\cB)$, we go over all $\{1\}\subseteq I\subseteq{}[\ell]$ and for each $I$, we evaluate $\det{{\cal T}({\cal B}[\cQ(I)])}$ and $\det{{\cal T}({\cal B}[\overline{\cQ(I)}])}$ in polynomial time via standard polynomial time determinant computation. Once this value is computed, we multiply their product with the evaluation of the monomial $m_I$. Since we go over $2^\ell$ possible sets $I$ and for each $I$ the computation takes polynomial time, the claimed running time follows. \qed
\end{proof}
Having shown that this polynomial can be efficiently evaluated, we will now turn to the way we use it in our algorithm. Our algorithm for {\sc $\exists$-Partition Connector} takes as input $G,f,\cQ$, evaluates the polynomial $P_\cQ(\bar x)$ at points chosen independently and uniformly at random from a field $\mathbb F$ of size $\Omega(n^6)$ and characteristic 2 and returns {\sc Yes} if and only if the polynomial does not vanish at the chosen points. In what follows we will prove certain properties of this polynomial which will be used in the formal proof of correctness of this algorithm.
%
%
%
We need another definition before we can state the main lemma capturing the properties of the polynomial. 
Recall that for every $v\in V(G)$, the set $A(v)$ is the set of `copies' of $v$ in the $f$-blowup of $G$. Furthermore, for a set $X\subseteq V(G)$, we say that an edge $e\in E(G)$ \emph{crosses} the cut $(X,\overline{X})$ if $e$ has exactly one endpoint in $X$.

\begin{definition}\label{def:monomial_editing} 
  We say that an $f$-factor $H$ of $G$ \emph{contributes} a monomial $x_{i_1j_1}^2\dots
  x_{i_rj_r}^2$ to the na{i}ve expansion of the right-hand side of
  \autoref{eqn:connectivity_polynomial_editing} if and only if the
  following conditions hold. 
  \begin{enumerate}
  \item For every $e=(v,w)\in E(H)$, there is a $u\in A(v)$, $u'\in A(w)$ and
    $1\leq p,q\leq r$ such that $\{u,v_e\}=\{i_p,j_p\}$ and
    $\{u',w_e\}=\{i_q,j_q\}$.

  \item For every $e=(v,w)\in E(G)\setminus E(H)$, there is a $1\leq p\leq r$ such
    that $\{v_e,w_e\}=\{i_p,j_p\}$.

  \item For every $1\leq p,q\leq r$, if $\{u,v_e\}=\{i_p,j_p\}$
    and $\{u',w_e\}=\{i_q,j_q\}$ for some $e\in E(G)$, then $e\in E(H)$.

  \item For every $1\leq p\leq r$, if $\{i_p,j_p\}=\{v_e,w_e\}$
    for some $e\in E(G)$, then $e\notin E(H)$.

  \item For every $1\in I\subseteq [\ell]$ such that $H$ has no edge crossing the cut $(\cQ(I),\overline {\cQ(I)})$, there is a pair of monomials $m_1$
    and $m_2$ such that $m_1$ is a surviving monomial in the
    simplified expansion of $ \det{\cT(\cB[\cQ(I)])}$, $m_2$ is a
    surviving monomial in the simplified expansion of
    $\det{\cT(\cB[\overline{\cQ(I)}])} $, and $m_1\cdot m_2\cdot
    m_I= x_{i_1j_1}^2\dots x_{i_rj_r}^2$.

\end{enumerate}

\end{definition}

Having set up the required notation, we now state the main lemma
which allows us to show that monomials contributed by
$f$-factors that do not connect $\cQ$, do not survive in the simplified
expansion of the right hand side
of~\autoref{eqn:connectivity_polynomial_editing}.

\begin{lemma}
\label{lem:monomial_editing}\label{lem:disconnected_factors_appear_even_times_editing} Every monomial in the polynomial $P_\cQ(\bar x)$ which is a surviving monomial in the simplified expansion of the right-hand side of \autoref{eqn:connectivity_polynomial_editing} is contributed by an $f$-factor of $G$ to the naive expansion of the right-hand size of  \autoref{eqn:connectivity_polynomial_editing}. Furthermore, for any $f$-factor  of $G$, say $H$, the following statements hold.
  \begin{enumerate}
  \item If
  \(H\) does not connect $\cQ$ then
  every monomial contributed by $H$  occurs an \emph{even} number of
  times in the polynomial $P_\cQ(\bar x)$ in the na{i}ve expansion of the right-hand side of \autoref{eqn:connectivity_polynomial_editing}.
  
  \item If \(H\)  connects $\cQ$, then every monomial contributed by $H$ occurs exactly once in the polynomial $P_\cQ(\bar x)$ in the na{i}ve expansion of the right-hand side of \autoref{eqn:connectivity_polynomial_editing}.

  \end{enumerate}
\end{lemma}

\begin{proof}
For the first statement, let $m$ be a monomial which survives in the simplified expansion of the right-hand side of \autoref{eqn:connectivity_polynomial_editing}. Then it must be of the form $x_{i_1j_1}^2\dots
  x_{i_rj_r}^2$ and must correspond to a perfect matching of $\cT(\cB)$. This is a direct consequence of Lemma \ref{lem:perfect_matching_determinant}~(2). Let $M$ be this perfect matching. We now define an $f$-factor $H$ based on $M$ and argue that $H$ indeed contributes this monomial $m$ to the naive expansion of the right-hand size of  \autoref{eqn:connectivity_polynomial_editing} as per Definition \ref{def:monomial_editing}. The $f$-factor $H$ is defined as follows. An edge $(v,w)\in E(G)$ is in $H$ if and only if the edge $(v_e,w_e)\notin M$. We now argue that $H$ contributes $m$. 
  
  Consider the first condition in Definition \ref{def:monomial_editing}. Since $e=(v,w)\in E(H)$, it must be the case that $(v_e,w_e)\notin M$. Since $M$ is a perfect matching and the vertices $v_e$ and $w_e$ each have exactly one neighbor other than each other, it must be the case that $M$ contains edges $e_1$ and $e_2$ where $e_1=(u,v_e)$ for some $u\in A(v)$ and $e_2=(u',w_e)$ for some $u'\in A(w)$. The fact that the second condition is satisfied follows directly from the definition of $H$. For the third condition, suppose that for some $1\leq p,q\leq r$, and $e=(u,v)\in E(G)$, it holds that $\{u,v_e\}=\{i_p,j_p\}$ and $\{u',w_e\}=\{i_q,j_q\}$. The fact that $M$ corresponds to $m$ implies that the edges $(u,v_e)$ and $(u',w_e)$ are in $M$, which in turn implies that the edge $(v_e,w_e)$ is \emph{not} in $M$. Hence, by definition of $H$, we conclude that $e\in E(H)$. An analogous argument implies that the fourth condition is satisfied as well. We now come to the final condition. Suppose that $1\in I\subseteq[\ell]$ such that $H$ has no edge crossing the cut $(\cQ(I),\overline{\cQ(I)})$.  Now, observe that for every $(v,w)\in E(G)$ which crosses the cut $(\cQ(I),\overline{\cQ(I)})$ the edge $e\notin E(H)$, which by definition implies that  $(v_e,w_e)\in M$. We define $\hat M$ to be the subset of edges $(v_e,w_e)\in M$ which cross the cut $(\cQ(I),\overline{\cQ(I)})$. Hence, for every edge $(v_e,w_e)$ in $M\setminus \hat M$, the vertices $v$ and $w$ lie on the same side of the cut $(\cQ(I),\overline{\cQ(I)})$. We now define a partition $M'\uplus M''$ of $M\setminus \hat M$ as follows. For $v\in V(G)$ and $u\in V(\cB)$, an edge $(u,v_e)\in M$ is in $M'$ if and only if $v\in \cQ(I)$. Clearly, $M'\uplus M''\uplus \hat M$ is now a partition of $M$. Furthermore, it is easy to see that $M'$ is a perfect matching of $\cT(\cB[\cQ(I)])$, $M''$ is a perfect matching of $\cT(\cB[\overline{\cQ(I)]})$. 
   
  Due to Proposition \ref{prop:perfect_matching_determinant}, we know that $M'$ corresponds to a surviving monomial $m'$ in the simplified expansion of $\det \cT(\cB[\cQ(I)])$ and $M''$ corresponds to a surviving monomial $m''$ in the simplified expansion of $\det \cT(\cB[\overline{\cQ(I)}])$. Finally, let $\hat m$ denote the monomial \(\prod_{(i_{k},j_{k})\in{}\hat M}x_{i_{k}{}j_{k}}^2\). It is easy to see that $m=m'\cdot m''\cdot \hat m$. Furthermore, $\hat m=m_I$. Hence we conclude that $m$ is indeed contributed by $H$ and proceed to the remaining two statements of the lemma.  However, before we prove the remaining statements, we  need the following claim.
 
 \begin{claim} Let $1\in I\subseteq [\ell]$.
 \begin{enumerate}
 \item If there is no edge of $H$ crossing the cut $(\cQ(I),\overline{\cQ(I)})$, then each monomial contributed by  $H$ to the naive expansion of the polynomial
   $\det \cT(\cB[\cQ(I)])\cdot \det \cT(\cB[\overline{\cQ(I)}])
   \cdot m_I$ is contributed exactly once.
 
 \item If  there is an edge of $H$ crossing the cut $(\cQ(I),\overline{\cQ(I)})$ then $H$  does not
   contribute a monomial to the naive expansion of the polynomial $\det
   \cT(\cB[\cQ(I)])\cdot \det \cT(\cB[\overline{\cQ(I)}])\cdot
   m_I$.
   \end{enumerate}

 \end{claim}
 
 \begin{proof}
 We begin with the proof of the first statement. By Definition \ref{def:monomial_editing} it holds that every monomial contributed by $H$ contains $m_I$. Let $H'$ be the subgraph of $H$ induced on $\cQ(I)$ and let $H''$ be the subgraph of $H$ induced on $\overline{\cQ(I)}$. Observe that $H'$ is an $f$-factor of $G[\cQ(I)]$ and $H''$ is an $f$-factor of $G[\overline{\cQ(I)}]$. By Proposition \ref{prop:perfect_matching_determinant} and Lemma \ref{lem:matching_equivalence}, we know that every $f$-factor of $G[\cQ(I)]$ ($G[\overline{\cQ(I)}]$) appears exactly once in the naive expansion of $\det \cT(\cB[\cQ(I)])$ ($\det \cT(\cB[\overline{\cQ(I)}])$) (since it is nothing but a perfect matching of the $f$-blowup induced by $\cQ(I)$ or $\overline{\cQ(I)}$). 
 
 Therefore, each monomial corresponding to a perfect matching of $\cB[\cQ(I)]$ which is equivalent to $H'$ appears exactly once in the naive expansion of the polynomial $\det \cT(\cB[{\cQ(I)}])$; similarly, each monomial corresponding to a perfect matching of $\cB[\overline{\cQ(I)}]$ which is equivalent to $H''$ appears exactly once in the naive expansion of $\det \cT(\cB[\overline{\cQ(I)}])$. Since every 
  monomial contributed by $H$ to the naive expansion of $\det \cT(\cB[\cQ(I)])\cdot \det 
  \cT(\cB[\overline{\cQ(I)}]) \cdot m_I$ is a product of $m_I$ and a monomial each from $\det \cT(\cB[{\cQ(I)}])$ and $\det \cT(\cB[\overline{\cQ(I)}])$, and these monomials themselves occur exactly once in the naive expansion of $\det \cT(\cB[{\cQ(I)}])$ and $\det \cT(\cB[\overline{\cQ(I)}])$ respectively, the first statement follows.

 We now prove the second statement of the claim. 
 Here, there must be vertices
 $v,w\in V(G)$ such that
 $v\in \cQ(I)$, $w\in \overline{\cQ(I)}$ and $(v,w)\in H$. Therefore, by
 Definition \ref{def:monomial_editing}, we have that no monomial
 contributed by $H$ has the term $x_{jk}^2$ where
 $\{j,k\}=\{v_e,w_e\}$. However, $m_I$
 contains the term $x_{jk}^2$ by definition. Therefore,
 $H$ does not contribute a monomial to $\det
 \cT(\cB[{\cQ(I)}])\cdot \det \cT(\cB[\overline{\cQ(I)}])\cdot
 m_I$. This completes the proof of the claim.
\qed \end{proof}

Let $\alpha$ be the number of connected components of the graph $H/\cQ$. If $H$ is an $f$-factor of $G$ that does not connect $\cQ$ it must be the case that $\alpha>1$. Due to the above claim, observe that there are exactly $2^\alpha$ sets $I$ such that $H$ contributes each of its monomials exactly once to the simplified expansion of the right hand side
of~\autoref{eqn:connectivity_polynomial_editing} and 
   $H$ does not contributes any monomials to any other sets $I$. Since $2^\alpha$ is even for $\alpha\geq 1$, we conclude that Statement 1 holds.

 We now move on to Statement 2. That is, we assume that $H$ is an $f$-factor that connects $\cQ$. Due to the above claim, we know that $H$ does not contribute a monomial to any polynomial $\det \cT(\cB[\cQ(I)])\cdot \det \cT(\cB[\overline{\cQ(I)}]) \cdot m_I$ where $1\in I\subset [\ell]$ is such that $H$ has an edge which crosses the cut $(\cQ(I),\overline{\cQ(I)})$. However, since $H$ connects $\cQ$, it crosses 
     \emph{every} $(\cQ(I),\overline{\cQ(I)})$ cut where $1\in I\subset [\ell]$. But observe that since $H$ is an $f$-factor of $G$ it will contribute a monomial to the polynomial $\det \cT(\cB[\cQ(I)])\cdot \det \cT(\cB[\overline{\cQ(I)}]) \cdot m_I$ when $I=[\ell]$. Hence, we conclude that any monomial contributed by $H$ occurs exactly once in the na{i}ve expansion of the right-hand side of \autoref{eqn:connectivity_polynomial_editing}, completing the proof of the lemma. 
\qed \end{proof}

This implies the following result, which is the last ingredient we need to prove Lemma~\ref{lem:existfactor}.

\begin{lemma}\label{lem:non-zero_coefficient} The polynomial $P_\cQ(\bar x)$ is not identically zero over ${\mathbb F}$ if and only if $G$ has an $f$-factor connecting $\cQ$.
\end{lemma}

\begin{proof}[Lemma \ref{lem:existfactor}] It follows from the definition of $P(\bar x)$ that its degree is $\bigoh(n^4)$ since  the number of vertices in the $f$-blowup of $G$ is $\bigoh(n^2)$. As mentioned earlier, our algorithm for {\sc $\exists$-Partition Connector} takes as input $G,f,\cQ$, evaluates the polynomial $P_\cQ(\bar x)$ at points chosen independently and uniformly at random from a field $\mathbb F$ of size $\Omega(n^6)$ and characteristic 2 and returns {\sc Yes} if and only if the polynomial does not vanish at the chosen points. Due to Lemma \ref{lem:non-zero_coefficient}, we know that the polynomial $P_\cQ(\bar x)$ is identically zero if and only if $G$ has an $f$-factor containing $\cQ$ and by the Schwartz-Zippel Lemma, the probability that the polynomial is not identically zero and still vanishes upon evaluation is at most $\frac{1}{n^2}$. This completes the proof of the lemma.
\qed \end{proof}

Having obtained the algorithm for \textsc{$\exists$-Partition Connector}, we now return to the algorithm for the computational version, \textsc{Partition Connector}.

\subsection{Solving \textsc{Partition Connector} in Randomized Polynomial Time}

\begin{proof}[Lemma \ref{thm:rand-part-conn}]
Consider the following algorithm $\cA$.  Algorithm $\cA$ takes as input an $n$-vertex instance of   \textsc{Partition Connector} with the partition $\cQ=\{Q_1,\dots, Q_\ell\}$, along with a separate set of edges $F$   that have been previously selected to be included in the partition connector. Let $F$ be initialized as $\emptyset$. As its first step, Algorithm $\cA$ checks if $\ell=1$; if this is the case, then it computes an arbitrary $f$-factor $H$, and outputs $H\cup F$. To proceed, let us denote the algorithm of Lemma~\ref{lem:existfactor} as $\cA'$. If $\ell>1$, then $\cA$ first calls $\cA'$ and outputs {\sc NO} if $\cA'$ outputs {\sc NO}. Otherwise, it fixes an arbitrary ordering $E^\leq$ of the edge set $E$ and recursively proceeds as follows.

$\cA$ constructs the set $E_1$ of all edges with precisely one endpoint in $Q_1$, and loops over all edges in $E_1$ (in the ordering given by $E^\leq$). For each processed edge $e=(v,w)$ between $Q_1$ and some $Q_i$ ($i\neq 1$), it will compute a subinstance $(G^e, f^e, \cQ^e)$ defined by setting: 
\begin{itemize}
\item $G^e=G-e$, and 
\item $f^e(v)=f(v)-1$, $f^e(w)=f(w)-1$ and $f^e=f$ for all the remaining vertices of $G$, and
\item $\cQ^e$ is obtained from $\cQ$ by merging $Q_1$ and $Q_i$ into a new set; formally (assuming $i<\ell)$, $\cQ^e=\{Q^e_1=Q_1\cup Q_i, Q_2,\dots, Q_{i-1}, Q_{i+1},\dots,
Q_\ell\}$.
\end{itemize}
Intuitively, each such new instance corresponds to forcing the $f$-factor to choose the edge $e$.
$\cA$ then queries $\cA'$ on $(G^e, f^e, \cQ^e)$. If $\cA'$ answers {\sc NO} for each such tuple $(G^e, f^e, \cQ^e)$ obtained from each edge $e$ in $E_1$, then $\cA$ immediately terminates and answers {\sc NO}. Otherwise, let $e$ be the first edge where $\cA'$ answered {\sc YES}; then $\cA$  adds $e$ into $F$. If $|\cQ^e|=1$ then the algorithm computes an arbitrary $f$-factor $H$ of $(G^e, f^e)$ and outputs $H\cup F$. On the other hand, if $|\cQ^e|>1$ then $\cA$ restarts the recursive procedure with $(G, f, \cQ):=(G^e, f^e, \cQ^e)$; observe that $|\cQ^e|\leq |\cQ|-1$.

Before arguing correctness, we show that the algorithm runs in the
required time. Since each edge in the partitioning is processed at
most $\ell$ times, the runtime of $\cA$ is asymptotically
upper-bounded by its at most $\ell\cdot n^2\leq n^3$ many calls to
$\cA'$. From Lemma~\ref{lem:existfactor}, we then conclude that the
total runtime of ${\cA}(G,f,\cQ)$ is upper-bounded by $2^{\ell}\cdot
n^{\bigoh(1)}$.

For correctness, let us first consider the hypothetical situation
where $\cA'$ always answers correctly. If no partition connector
exists, then $\cA$ correctly outputs {\sc NO} after the first call
to $\cA'$. Otherwise, there exists a partition connector, and such a
partition connector must contain at least one edge in $E_1$ at every recursion of the algorithm. This implies that $\cA'$ would output {\sc YES}
for at least one edge $e$ of $E_1$. Moreover, it is easily seen that
for any partition connector $T$ containing $e$, $T\setminus \{e\}$ is also
a partition connector in $(G^e, f^e, \cQ^e)$, and so by the same
argument $\cA'$ would also output {\sc YES} for at least one edge in the individual sets $E_1$ constructed in the recursive calls of $\cA$. In particular, if $\cA'$ would always answer correctly, then $\cA$ would correctly output a partition connector $H\cup F$ at the end of its run. For further considerations, let us fix the set $F$ which would be computed by $\cA$ under the assumption that $\cA'$ always answers correctly; in other words, $F$ is the lexicographically first tuple of edges in $E_1$ which intersects a partition connector.

We are now ready to argue that $\cA$ succeeds with the desired
probability; recall that $\cA'$ only allows one-sided errors. So, if
the input is a no-instance, then $\cA$ is guaranteed to correctly
output {\sc NO} after the first query to $\cA'$. Furthermore, by the
definition of $F$, for each edge $e\not \in F$ processed by $\cA$, the
algorithm $\cA'$ must also answer {\sc NO} on $(G^e, f^e, \cQ^e)$. So, assuming $\cA'$ always answers correctly, in
total $\cA'$ would only be called at most times on
yes-instances, and in all remaining calls it receives a
no-instance. Given that $\cA'$ has a success probability of at least
$1-\frac{1}{n^2}$, the probability that $\cA'$ is called at most $|F|+1=\ell$ times on  {\sc YES}-instances (not counting the initial call on $G$), and that it succeeds in all these calls,
is at least $(1-\frac{1}{n^2})^\ell$. Hence the error probability of
the algorithm is at most $1-(1-\frac{1}{n^2})^\ell$. This completes
the proof of the lemma.
\qed \end{proof}

\section{Classification Results}
\label{sec:npi}

In this section, we prove Theorem \ref{thm:lowerbound} which we restate for the sake of completeness.

\lowerbound*

The result relies on the established Exponential Time Hypothesis, which we recall below.
\begin{definition}[Exponential Time Hypothesis (ETH),~\cite{RF01}]
There exists a constant $s>0$ such that \textsc{3-SAT} with $n$ variables and $m$ clauses cannot be solved in time $2^{sn}(n+m)^{\mathcal{O}(1)}$.
\end{definition}

We first show that the problem is not \NP-hard unless the ETH fails. We remark that we can actually prove a stronger statement here by weakening the premise to ``\NP\  is not contained in Quasi-Polynomial Time''. However, since we are only able to show the other part of Theorem \ref{thm:lowerbound} under the ETH, we phrase the statement in this way.

\begin{lemma}\label{lem:non-npc}
For every $c>1$ and for every $g(n)\in\Theta((\log n)^c)$, \textsc{Connected $g$-Bounded $f$-Factor} is not \NP-hard unless the Exponential Time Hypothesis fails.
	\end{lemma}
	
	\begin{proof} Due to Theorem~\ref{thm:polylog}, we know that  when $g(n)\in\Theta((\log n)^c)$, \textsc{Connected $g$-Bounded $f$-Factor} can be solved in quasi-polynomial time. Hence, this problem cannot be \NP-hard unless \NP\  is contained in the complexity-class Quasi-Polynomial Time, \QP. Furthermore, observe that \NP$\subseteq$ \QP\ implies that the ETH is false. Hence, we conclude that \textsc{Connected $g$-Bounded $f$-Factor} is not \NP-hard unless the Exponential Time Hypothesis fails.
  \qed \end{proof}


Next, we use a reduction from \textsc{Hamiltonian Cycle} to obtain:
\begin{lemma}
\label{lem:non-p}
	For every $c>1$ and for every $g(n)\in\Theta((\log n)^c)$, \textsc{Connected $g$-Bounded $f$-Factor} is not in \P\ unless the Exponential Time Hypothesis fails.
\end{lemma} 

\begin{proof} 
Assume for a contradiction that \textsc{Connected $g$-Bounded $f$-Factor} is in \P\ for some $g(n)\in\Theta((\log n)^c)$ and $c>1$. Let us fix this function $g$ for the remainder of the proof. In particular, there exists constants $c_1$ and $\epsilon>0$ such that $g(n)\geq c_1(\log n)^{1+\epsilon}$ for sufficiently large $n$.
The proof is structured as follows. First, we present a \emph{subexponential time} reduction from \textsc{Hamiltonian Cycle} to \textsc{Connected $g$-Bounded $f$-Factor}. We then show that such a reduction would imply a subexponential time algorithm for \textsc{Hamiltonian Cycle}, which is known to violate ETH~\cite{RF01}.

The reduction algorithm $\textsc{R}_\epsilon$ takes a graph $G$ on $z$ vertices as input, computes $s=\frac{2^{(\frac z{c_1})^{1/(1+\epsilon)}}}{z}$, and outputs an $n$-vertex instance $(G',f)$ of \textsc{Connected $g$-Bounded $f$-Factor} which satisfies the following conditions:
\begin{enumerate}
\item $f(v)\geq\frac{n}{c_1(\log n)^{1+\epsilon}}$ for every $v$ in $G'$.
\item $n$ is upper-bounded by a subexponential function of $z$.
\item $G$ has a Hamiltonian cycle if and only if $(G',f)$ contains a connected $f$-factor. 
\end{enumerate}
Crucially, observe that for sufficiently large $z$, we have $s>z$.
The algorithm $\textsc{R}_\epsilon$ works as follows. Given a graph $G$, for each vertex $v$ it constructs a clique $C_v$ of size $\ceil{s}-1$ and makes each vertex in $C_v$ adjacent to $v$. For each $x\in C_v$, it sets $f(x)=\ceil{s}-1$, while for $v$ it sets $f(v)=\ceil{s}+1$.

Next, we argue that $(G',f)$ satisfies conditions~$(1)$,~$(2)$ and~$(3)$. For Condition~$(1)$, we need to ensure that the bound on $f$ holds for vertices in each $C_v$, meaning that we need to verify that $\ceil{s}-1\geq\frac{n}{c_1(\log n)^{1+\epsilon}}=\frac{\ceil{s}\cdot z}{c_1(\log (\ceil{s}\cdot z))^{1+\epsilon}}$ holds. By replacing $s$ with its function of $z$, we obtain:

$$
\frac{\ceil{s}\cdot z}{c_1(\log (\ceil{s}\cdot z))^{1+\epsilon}}= \frac{\ceil{s}\cdot z}{c_1(\log \ceil{\frac{2^{(\frac z{c_1})^{1/(1+\epsilon)}}}{z}}\cdot z)^{1+\epsilon}}$$

We proceed by using $s>z$ to bound the effects of rounding up the numerator in the fraction.

$$
\frac{\ceil{s}\cdot z}{c_1(\log (\frac{2^{(\frac z{c_1})^{1/(1+\epsilon)}}}{z}\cdot z))^{1+\epsilon}}\geq 
\frac{\ceil{s}\cdot z}{c_1(\log (\ceil{\frac{2^{(\frac z{c_1})^{1/(1+\epsilon)}}}{z}}\cdot z))^{1+\epsilon}}
$$

It remains to show that $\ceil{s}-1$ is at least the left expression.

$$\frac{\ceil{s}\cdot z}{c_1(\log (\frac{2^{(\frac z{c_1})^{1/(1+\epsilon)}}}{z}\cdot z))^{1+\epsilon}}
=\frac{\ceil{s}\cdot z}{z\cdot (\log 2)^{1+\epsilon}}= \frac{\ceil{s}}{(\log 2)^{1+\epsilon}}.
$$

Since $\ceil{s}-1\geq\frac{\ceil{s}}{(\log 2)^{1+\epsilon}}$, Condition~$(1)$ holds. For Condition~$(2)$, it suffices to note that $n=\ceil{s}\cdot z= z\cdot \ceil{\frac{2^{(\frac z{c_1})^{1/(1+\epsilon)}}}{z}}$, which is clearly a subexponential function.

Finally, for Condition~$(3)$, observe that every edge in each clique $C_v$ must be used in every connected $f$-factor of $(G',f)$. Furthermore, all the other edges in every such connected $f$-factor must induce a connected subgraph of $G$ with degree $2$, which is a Hamiltonian cycle. Hence there is a one-to-one correspondence between Hamiltonian cycles in $G$ and connected $f$-factors of $(G',f)$, and Condition~$(3)$ also holds.

To complete the proof, recall that we assumed that there exists a polynomial time algorithm for \textsc{Connected $g$-Bounded $f$-Factor} for our choice of $g$. Then, given an instance $G$ of \textsc{Hamiltonian Cycle}, we can apply $\textsc{R}_\epsilon$ on $G$ followed by the hypothetical polynomial time algorithm on the resulting instance $(G',f)$ (whose size is subexponential in $|V(G)|$) to solve $G$ in subexponential time. As was mentioned earlier in the proof, such an algorithm would violate ETH.
\qed \end{proof}

Lemmas \ref{lem:non-npc} and \ref{lem:non-p} together give us Theorem
\ref{thm:lowerbound}.

\section{Concluding remarks}

We obtained new complexity results for \textsc{Connected $f$-Factor}
with respect to lower bounds on the function $f$.
As our main results, we showed that when $f(v)$ is required to be at least
$\frac{n}{(\log n)^c}$,  the problem can be solved in quasi-polynomial
time in general and in randomized polynomial time if $c\leq
1$. Consequently, we show that the problem can be solved in
polynomial time when $f(v)$ is at least $\frac{n}{c}$ for any constant
$c$. We complement the picture with matching classification results.

As a by-product we obtain a generic approach reducing
\textsc{Connected $f$-Factor} to the ``simpler'' \textsc{Partition
  Connector} problem. Hence future algorithmic improvements of \textsc{Partition
  Connector} carry over to the \textsc{Connected $f$-Factor} problem.
Finally, it would be interesting to investigate the possibility of
derandomizing the polynomial time algorithm for the case when $g(n)
= \mathcal{O}(\log n)$.

\paragraph*{Acknowledgments.} \hspace{0.05cm} The authors wish to thank the anonymous reviewers for their helpful comments. The authors acknowledge support by the Austrian Science Fund (FWF, project P26696), and project TOTAL funded by the European Research Coun-
cil (ERC) under the European Unions Horizon 2020 research and innovation programme (grant agreement No 677651). Robert Ganian is also affiliated with FI MU, Brno, Czech Republic.

\bibliographystyle{plain}
\bibliography{biblio}

\end{document}